\documentclass[pdflatex,sn-mathphys-num]{sn-jnl}

\usepackage{scalerel}
\usepackage{tikz}
\usetikzlibrary{svg.path}

\definecolor{orcidlogocol}{HTML}{A6CE39}
\tikzset{
  orcidlogo/.pic={
    \fill[orcidlogocol] svg{M256,128c0,70.7-57.3,128-128,128C57.3,256,0,198.7,0,128C0,57.3,57.3,0,128,0C198.7,0,256,57.3,256,128z};
    \fill[white] svg{M86.3,186.2H70.9V79.1h15.4v48.4V186.2z}
                 svg{M108.9,79.1h41.6c39.6,0,57,28.3,57,53.6c0,27.5-21.5,53.6-56.8,53.6h-41.8V79.1z M124.3,172.4h24.5c34.9,0,42.9-26.5,42.9-39.7c0-21.5-13.7-39.7-43.7-39.7h-23.7V172.4z}
                 svg{M88.7,56.8c0,5.5-4.5,10.1-10.1,10.1c-5.6,0-10.1-4.6-10.1-10.1c0-5.6,4.5-10.1,10.1-10.1C84.2,46.7,88.7,51.3,88.7,56.8z};
  }
}

\newcommand\orcidicon[1]{\href{https://orcid.org/#1}{\mbox{\scalerel*{
\begin{tikzpicture}[yscale=-1,transform shape]
\pic{orcidlogo};
\end{tikzpicture}
}{|}}}}

\usepackage{graphicx}%
\usepackage{multirow}%
\usepackage{amsmath,amssymb,amsfonts}%
\usepackage{amsthm}%
\usepackage{mathrsfs}%
\usepackage[title]{appendix}%
\usepackage{xcolor}%
\usepackage{textcomp}%
\usepackage{manyfoot}%
\usepackage{booktabs}%
\usepackage{algorithm}%
\usepackage{algorithmic}
\usepackage{listings}%
\usepackage{babel}
\usepackage{comment}
\usepackage{bm}
\usepackage{amsmath}
\usepackage{amsfonts}
\usepackage{latexsym}
\usepackage{booktabs}
\usepackage{subcaption}
\DeclareGraphicsExtensions{.jpg,.jpeg,.pdf,.png,.mps}
\usepackage{epsfig}
 
\usepackage{rotating}
\usepackage{color}
\usepackage{xcolor}
\usepackage{colortbl}
\usepackage[T1]{fontenc}
\usepackage{threeparttable}
\usepackage{lipsum}
\usepackage{url} 
\usepackage{hyperref}
\usepackage{lmodern}
\date{}
\hypersetup{hidelinks}

\theoremstyle{thmstyleone}%
\newtheorem{theorem}{Theorem}
\theoremstyle{thmstyletwo}%

\theoremstyle{thmstylethree}%
\newtheorem{definition}{Definition}%

\begin{document}

\title[Random Reward Phase-Type Distributions]{Random Reward Phase-Type Distributions\\With Applications in Latent Severity Modeling}


\author*{\fnm{Simon} \sur{Pauli} \orcidicon{0009-0004-7458-9775}}\email{simon.pauli@jku.at}

\author{\fnm{Andreas} \sur{Futschik} \orcidicon{0000-0002-7980-0304}}\email{andreas.futschik@jku.at}

\affil[1]{\orgdiv{Institute of Applied Statistics}, \orgname{Johannes Kepler University}, \orgaddress{\street{Altenberger Straße 69}, \city{Linz}, \postcode{4040}, \state{Upper Austria}, \country{Austria}}}
\newcommand{\logit}{\text{logit}}


\abstract{    This paper proposes an extension to discrete Phase-Type distributions (DPH) by introducing random rewards. These allow for modeling a system in which a visit to a certain state does not emit a deterministic reward. Instead, the rewards follow either a Bernoulli or a geometric distribution. Utilizing this increased flexibility, we further sketch a possible use case for these random rewards by introducing the Inertia-Escalation model (IEM), a process with latent severity levels characterized through two parameters: Inertia $\nu$ and escalation $\eta$. We also discuss parameter inference for such models.
    To validate and explore random rewards and the IEM, we conducted extensive simulations and applied the model to two datasets: historical warfare and the Telco customer churn dataset.}




\maketitle

\section{Introduction}\label{introSec}

Markov processes are widely used to model various natural, social, and business phenomena. If a Markov Chain has one absorbing state, the distribution of the time $\tau$ to absorption is called a Phase-Type distribution.

Consider a discrete-time Markov chain with  the transition matrix 
\begin{center}\begin{equation}\label{fullchainEq}
    \boldsymbol\Lambda=
    \bordermatrix{&\textbf{S}&S^{abs} \cr
         \textbf{S}&\textbf{T}&\textbf{t}  \cr
         S^{abs}&\textbf{0}&1 
         }.
\end{equation}\end{center}

In the \textit{sub-transition} matrix $\textbf{T}\in[0,1]^{d\times d}$, $\boldsymbol\Lambda$ contains transition probabilities between the transient states $\textbf{S}=(1,\ldots,d)=(S^1,\ldots, S^d)$ (the latter representation is sometimes used for clarity). Furthermore, there is one absorbing state, $S^{abs}$. The \textit{exit probabilities} represent the probability of leaving from the transient states into the absorbing state, and are given by $t=(\textbf{I}-\textbf{T})\textbf{e}$, where $e=(1,1,\ldots,1)^\top$.

Let furthermore $\pi\in[0,1]^d$ 
denote a vector of starting probabilities. According to  \cite{BladtBook}, the density of the number of transitions $\tau$ until reaching the absorbing state is given as
\begin{equation}\label{pdfDPH}
  f(n):= P(\tau=n) =
  \bm{\pi} \bm{T}^{n-1} \bm{t}, quad n\in\mathbb N_0
\end{equation}
 and denoted as $\tau\sim DPH(\pi,\textbf{T})$.

Sometimes, interest is not only in the time until absorption but rather in the time spent in different states during the process. That is, either only a subset of states is considered noteworthy, or the states have different levels of importance. One may introduce rewards to carry out more sophisticated counting schemes.
They may be coded by a diagonal reward matrix $\Delta$,  where the diagonal entry $\Delta_{ii}=r_i\in\mathbb N_0$ indicates the reward of the corresponding state $S^i$. For simplicity, we use $S^i$ and $i$ interchangeably.
Let $S_t$ be the state of the process at time $t$. Then $\tilde{\tau} = \sum\limits_{t=1}^\tau r_{S_t}$ is the sum of rewards of the visited states until absorption. 
If rewards are nonnegative integer-valued, it has been shown that  $\tilde{\tau}$ also follows a DPH distribution, see \cite{hobolth24}.

Several reward schemes can be considered simultaneously, leading to multivariate discrete Phase-Type (MDPH) distributions \cite{navarro}. 
We denote such as situation by $\textbf{Y}\sim MDPH(\boldsymbol\pi,\textbf{T},\textbf{R})$, with $\textbf{Y}\in\mathbb N_0^{n\times k}$ and $\textbf{R}\in(\mathbb N_0)^{d\times k}$, where $k$ is the number of different measurements. The entries $R_{ij}$ then denote the reward that state $i$ gets in terms of measurement $j$, and $R_{\cdot j}$ would correspond to the diagonal of the $j-th$ observations' reward matrix.

 The main contribution of this paper is to generalize the reward system by proposing \textit{random rewards}, which allow the rewards of a state to follow either a Bernoulli or a geometric distribution. At each visit to a state, a random quantity is drawn from the corresponding distribution, and the sum of these emissions is observed. Previously, rewards were treated as fixed quantities. We show that for suitable reward distributions, random rewards also lead to a DPH--distribution.
 We call the resulting distribution {\em random reward discrete phase type}, or RRDPH.
 Theoretical results and inference methods that already exist for DPH also apply to RRDPH.

 Using the EM algorithm proposed by \cite{He2016Sep}, both the reward probabilities of certain states as well as parameters describing the transition matrix \textbf{T} may be estimated. Our simulations confirm that the true parameters can be recovered given a sufficiently large sample size.

 We also consider the bivariate case where two variables are observed. One is the number of accumulated rewards, denoted as $\psi\in\mathbb N_0$. The other is the time $\tau\in\mathbb N$ the process took (i.e., the number of steps until the absorbing state is reached). We show that the tuple of $(\psi,\tau-\psi)$ follows an MDPH also under suitably distributed random rewards.

 Furthermore, we consider the Inertia-Escalation Model (IEM). Many real-life processes are modeled using distinct, discrete stages of escalation. Insurance claim frequency, for example, is often studied using Hidden Markov Models or other severity models (\cite{Gao2021Apr}, \cite{VERSCHUREN2022379}, \cite{Park2023Jul}). However, such models can also be applied in a variety of other situations, ranging from modeling the severity of a conflict (\cite{stoehr-etal-2023-ordinal}, \cite{williams24}) to psychometric latent, ordinal indicators, such as trustworthiness or frailty (\cite{Nguyen2014}, \cite{Wolff2024Apr}).

In such situations, different measurements may be observable. When modeling conflicts, one observes not only the time until the conflict ends but also the number of casualties or the number of ceasefires. Similarly, in claim frequencies, both the frequency and the severity are possible outcomes (\cite{Janousek2025Jun}). A multivariate model permits the investigation of dependencies in such situations.

We also consider parameter inference for random-reward IEMs using different observables. Both the model and the random rewards are extensively tested in simulations and applied to two datasets: Telco customer churn and a historical conflict data set.

This paper will first extend the DPH framework by considering random rewards in section \ref{rrSec}. This is followed by proposing the IEM in section \ref{iemSec}, a discussion of parameter inference in \ref{inferenceSec}, simulations supporting the validity of the inference in section \ref{simSec}, and an application to the famous Telco customer churn data and historical warfare data in section \ref{ApplSec}. Concluding remarks and a research outlook can be found in section \ref{discSec}.

\section{Random Rewards}\label{rrSec}

This section will introduce random rewards and discuss how they can be embedded into a discrete Phase-Type framework. Assuming a suitable reward distribution, such models inherit the convenient properties of DPH, including closed-form solutions for the likelihood and moments, as well as software implementations. 
The resulting Random Reward Discrete Phase-Type (RRDPH) distributions
will also be studied in a multivariate setting. Random rewards can be added to any DPH process (including the IEM proposed in section \ref{iemSec}).

In this paper, we show in particular that two reward
distributions permit us to stay in a DPH framework. Bernoulli rewards (see section \ref{berRewSec}), and geometric rewards (see section \ref{geoRewSec}). 

\subsection{Discrete Phase-Type Distributions with Bernoulli Rewards}\label{berRewSec}

To explain the usage of  Bernoulli rewards, consider the following simple example: We take four transient states, A, B, C, and D, where C and D are connected to the absorbing state as shown in Figure \ref{fig:exDPH} (a). While states A and B have a fixed reward of 1, state C has a reward that is Bernoulli distributed with parameter $p_C:=0.6$, and state D has a Bernoulli distributed reward with $p_D:=0.3$. The process starts in A and goes to B or C with probabilities $b:=\frac12$ and $1-b=\frac12$, respectively. From B, a transition to D takes place. Finally, from C and D, the absorbing state will be reached in the next step. The random reward means that if a transition from A to C occurs, the accumulated reward is 2 with probability 0.6 and 1 with probability 0.4. In the $A\rightarrow B$ case, the accumulated rewards will be 2 with probability 0.7 and 3 with probability 0.3. If all rewards are fixed to one, the reward-weighted absorption time would be 

\begin{equation}
\tilde{\tau}\sim DPH\left((1,0,0,0),\left(\begin{array}{cccc}
     0&\frac12&\frac12&0  \\
     0&0&0&1\\
     0&0&0&0\\0&0&0&0
\end{array}\right)\right).\end{equation}\label{exampRRDPH}

To model a Bernoulli reward for state C, we replace C by two distinct states: $C_0$, which corresponds to state C under the Bernoulli reward of 0, and $C_1$, where the reward is 1. Now states  A, B, and $C_1$ are assigned the fixed reward 1, whereas $C_0$ is assigned $0$. The  transition probability for $A\rightarrow C_1$ is $\hat{\textbf{T}}_{AC_1} = \textbf{T}_{AC}\times p_C = 0.3$. Furthermore $\hat{\textbf{T}}_{AC_0}  = \textbf{T}_{AC}*(1-p_C) = 0.2$.

The same can be done for D. The nonrewarded state, $D0$, is visited with $1-p_D=0.7$. Therefore, $D1$ is visited with $p_D=0.3$. By introducing these new states, we transform the chain with a Bernoulli reward into a fixed-reward chain. The resulting chain is displayed in Figure \ref{fig:exDPH} (b).
\begin{figure}
    \centering
    \begin{subfigure}{.5\textwidth}
        \centering
        \includegraphics[width=1.2\linewidth]{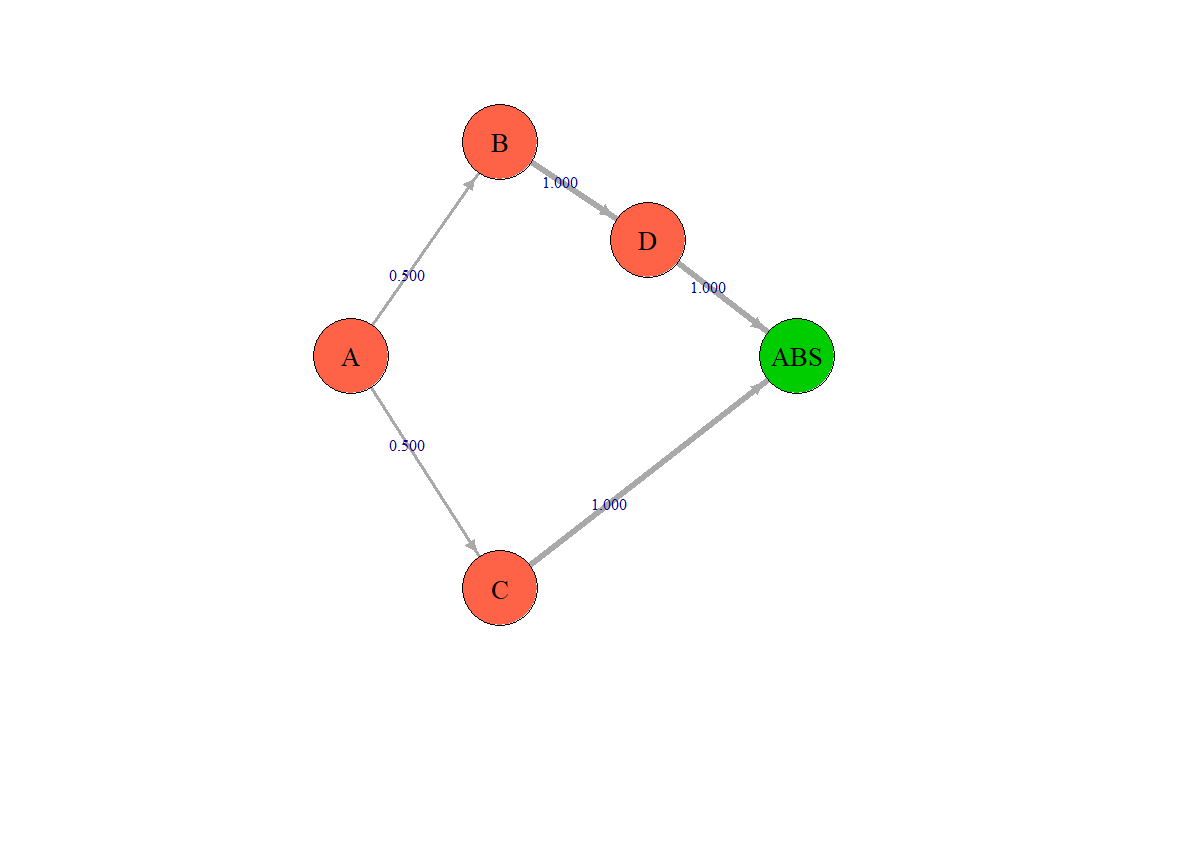}
        \caption{Untransformed DPH}

    \end{subfigure}%
        \begin{subfigure}{.5\textwidth}
    \centering
    \includegraphics[width=1.2\linewidth]{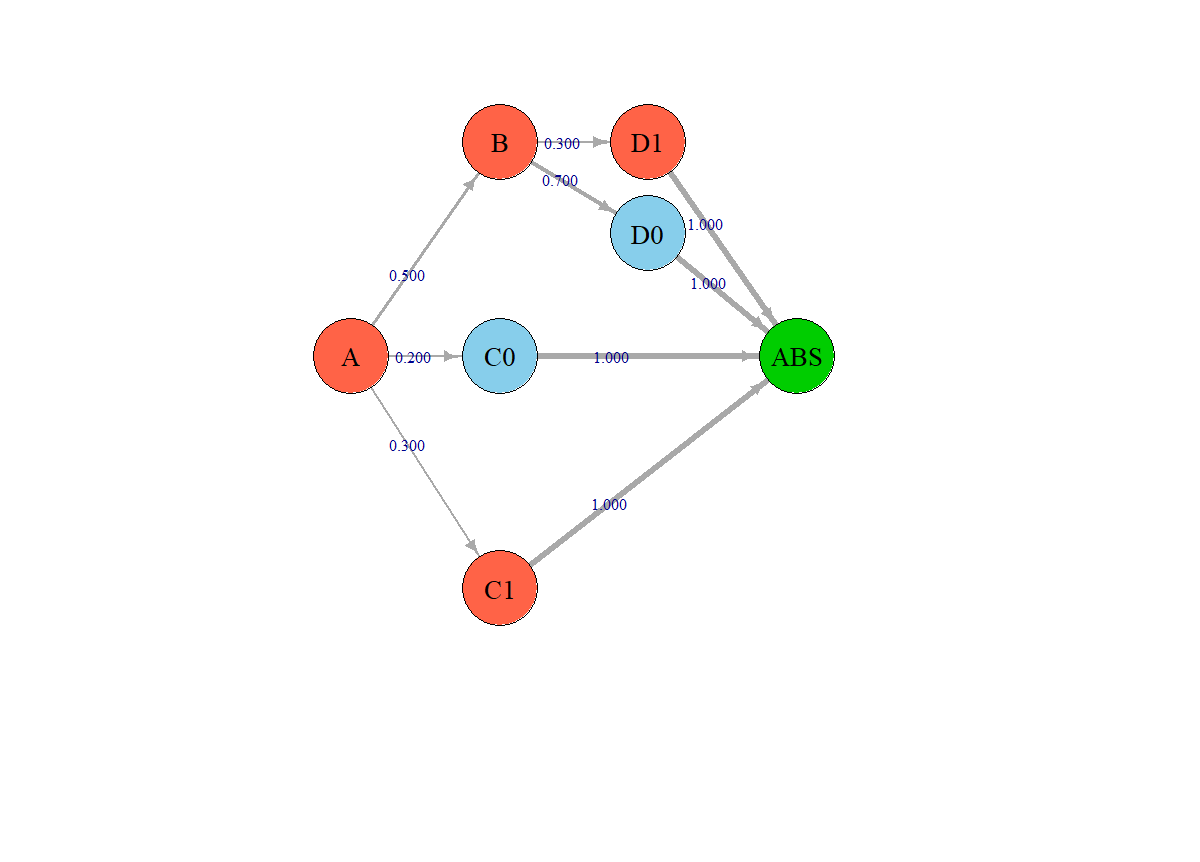}
    \caption{Transformed DPH}
    \label{berToyExNEW}
    \end{subfigure}
    \caption{Original stochastic system (a), followed by a split as $r(C)\sim Ber(0.6)$ and $r(D)\sim Ber(0.3)$ (b). Red states have a reward of 1; the blue states $C_0$ and $D_0$ have a reward of 0 in the Bernoulli reward version.}
    \label{fig:exDPH}
\end{figure}

Generalizing from this to a system where every state has Bernoulli rewards, let \textbf{T} be the subtransition matrix of a DPH and $\textbf{p}\in [0,1]^d$ be the probabilities of rewards of each stage. Further, let \textbf{P} be a diagonal matrix with diagonal entries $P_{ii}=p_i$, then define the extended matrix 
\begin{equation}\label{foldTrmatBer}
    B_b(\textbf{T},\textbf{P})=\left(\begin{array}{cc}
         \textbf{T}(\textbf{I}-\textbf{P}) &\textbf{T}\textbf{P} \\
         \textbf{T}(\textbf{I}-\textbf{P}) &\textbf{T}\textbf{P}
    \end{array}\right).
\end{equation}

The general idea is the same as in the previous example. All states are split into two: one with reward 1 and the other with reward 0. The transition probabilities are multiplied by the corresponding Bernoulli parameters. 
A state with a non-random reward receives a Bernoulli parameter $p_j=1$. 
Since a transition decision precedes the reward allocation, the transition probabilities of (splitted) states $j$ are multiplied with $p_j$ at the right-hand side of (\ref{foldTrmatBer})
and with $1-p_j$ at the left-hand side.
As the process may start in either a rewarded or an unrewarded state, the starting distribution is also altered to be  \begin{equation}\label{berNewIniEq}\boldsymbol\beta(\boldsymbol\pi,\textbf{P}) = ((\textbf{I}-\textbf{P})\boldsymbol{\pi},\textbf{P}\boldsymbol{\pi})^\top.\end{equation}

The accumulated rewards correspond to observations using the reward diagonal $\textbf{r}_1 = (\textbf{0}_d,\textbf{1}_d)^\top$. 
Thus, for each pair of new states produced by a split, one receives a reward of zero, and the other receives a reward of one.

We consider two potential observations, the number
of rewarded steps $\psi$, and the total number of steps
until absorption $\tau$. To count the total number of steps, all rewards are set to one. Alternatively, one may also count the non-rewarded steps $\tau-\psi$ using
an obvious choice of rewards.

Together, the observations $(\psi,\tau-\psi)$ may be used to estimate model parameters. The inference algorithms are presented in section \ref{inferenceSec}.

\begin{definition}\label{berRewDef}
    An observation $\textbf{Y}=(Y_1, Y_2)\in\mathbb N_0^2$ is said to come from a random reward DPH with Bernoulli rewards, denoted $\textbf{Y}\sim \mathrm{RRDPH}_B(\boldsymbol\pi,\textbf{T},\textbf{P},
    (\textbf{r}_1,\textbf{r}_2))$ if $\textbf{Y}\sim MDPH\left(\boldsymbol\beta(\boldsymbol\pi,\textbf{P}),\textbf{B}_b(\textbf{T},\textbf{P}),(\textbf{r}_1,\textbf{r}_2)\right)$, with $\boldsymbol\beta(\boldsymbol\pi,\textbf{P})$ as in \eqref{berNewIniEq} and $\textbf{B}_b(\textbf{T},\textbf{P})$ as in \eqref{foldTrmatBer}.
\end{definition}

Theorem \ref{berRewTheo} shows that $\mathrm{RRDPH}_B$ provides the distribution of accumulated Bernoulli rewards $(\psi,\tau-\psi)$ that are emitted by a Markov chain. The proof is in the appendix.

\begin{theorem}\label{berRewTheo}
Let $T\in[0,1]^{d\times d}$ be the sub-transition matrix of a DPH with initial probabilities $\pi\in\mathbb{R}^d$ and absorption vector $\textbf{t}=(\textbf{I}_d-\textbf{T})\textbf{e}$.  Let $\textbf{p}=(p_1,\dots,p_d)^\top\in[0,1]^d$ and $\textbf{P}=\text{diag}(p_1,\dots,p_d)$.  Consider a model where on each visit to state $i$ an independent $\mathrm{Bernoulli}(p_i)$ determines whether the visit is rewarded, and let $X=(\psi,\tau-\psi)$ denote the rewarded and unrewarded counts, respectively (so $\psi$ is the number of rewarded visits and $\tau$ the absorption time). Let $\textbf{Y}=(Y_1,Y_2)\sim \mathrm{RRDPH}_B(\boldsymbol\pi,\textbf{T},\textbf{P},(\textbf{r}_1,\textbf{r}_2))$ as in definition \ref{berRewDef}, and with
$\textbf{r}_1 = (\textbf{0}_d,\textbf{1}_d)^\top$,
and $\textbf{r}_2 = (\textbf{1}_d,\textbf{0}_d)^\top$.
Then \textbf{X} and \textbf{Y} have the same distribution.
\end{theorem}

\subsection{Discrete Phase-Type Distributions with Geometric Rewards}\label{geoRewSec}

Geometric rewards can be incorporated into a DPH framework by exploiting the memoryless property of the geometric distribution. In contrast to Bernoulli rewards, the reward accumulated in a state is not decided by a single success/failure draw, but by repeated Bernoulli trials until the first failure. This allows us to represent geometric rewards by an expanded state space while preserving the phase-type structure.

To motivate the expanded matrix, we consider a simple fixed reward model shown in Figure \ref{geoToyEx}. If a random reward is assigned to $C$,  with a geometric reward distribution $r(C)\sim Geom(q)$, this may be implemented
by splitting $ C$ into two states, $C_0$ and $C_1$ as shown in 
Figure \ref{geoSplitPic}.
Unlike the Bernoulli case, $C_0$ and $C_1$ are not mutually exclusive; Instead, reaching $C_0$ assigns a reward of zero. By moving to $C_1$, rewards are accumulated until this state is left. Thus, $P(A\rightarrow C_0)=P(A\rightarrow C)$, after which $P(C_0\rightarrow C_1)=1-q$. If $C_1$ is reached, the rewards keep accumulating (each transition happening with $P(C_1\rightarrow C_1)=1-q$) until a failure occurs
with probability $q$. In our example, failure corresponds to the transition $C_1\rightarrow ABS$. In summary, the geometric reward is represented by the number of repeated self-transitions before the process leaves the reward-accumulating state.

\begin{figure}
\begin{subfigure}{.5\textwidth}
        \centering
        \includegraphics[width=1.2\linewidth]{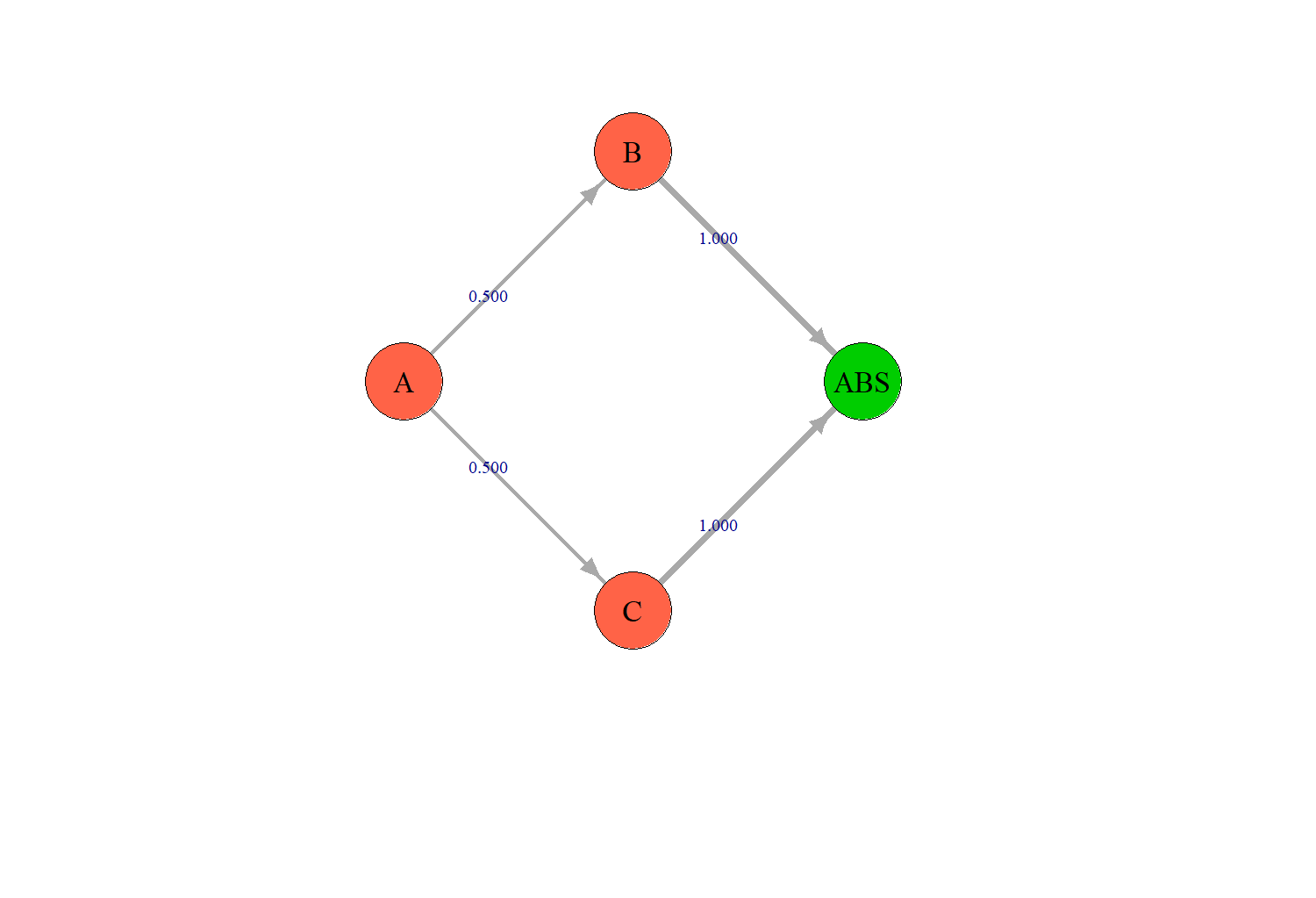}
        \caption{Untransformed DPH}
        \label{geoToyEx}
        
    \end{subfigure}%
        \begin{subfigure}{.5\textwidth}
    \centering
    \includegraphics[width=1.2\linewidth]{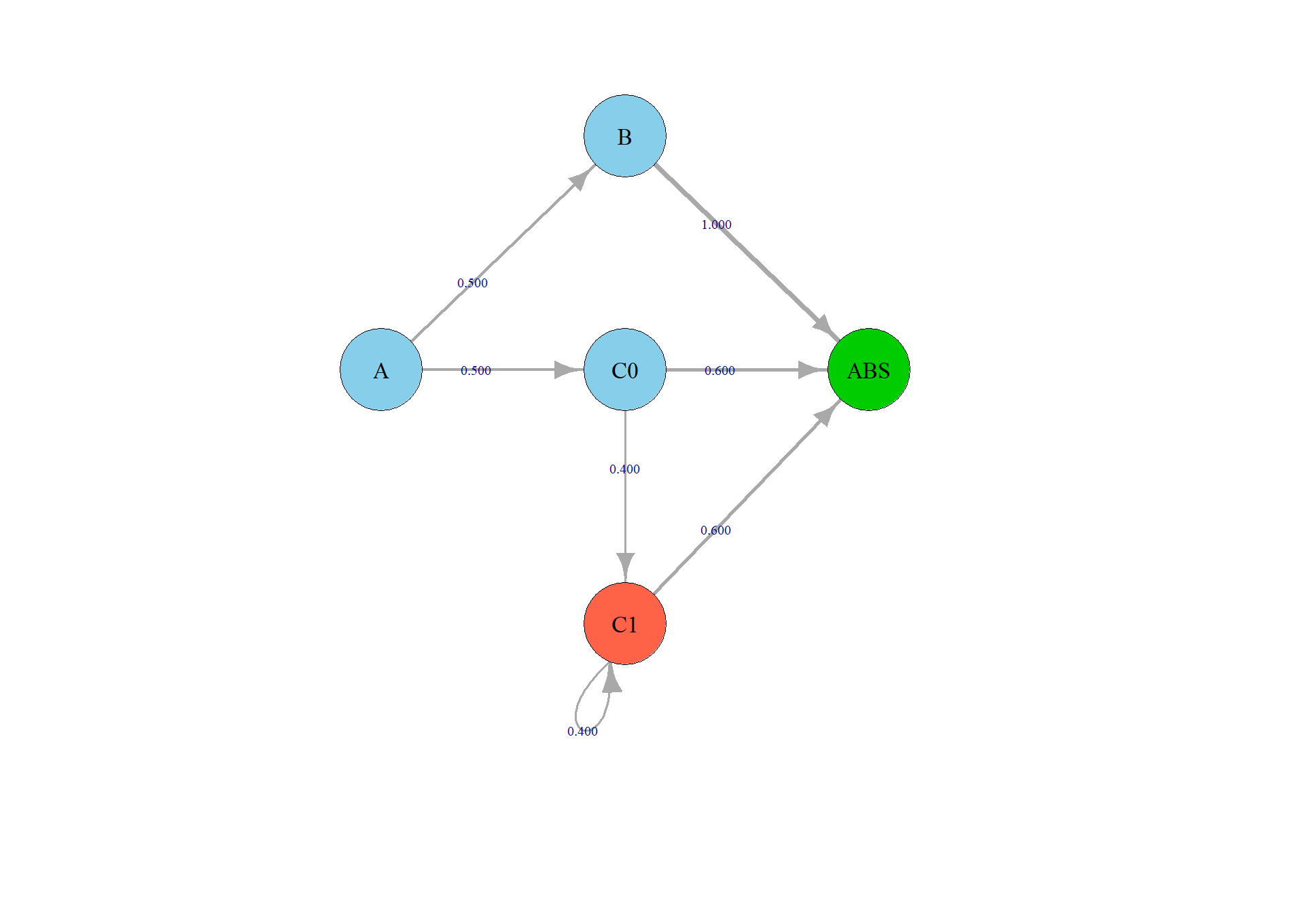}
    \caption{Split system as $r(C)\sim Geo(0.6)$, red state $C_1$ captures random rewards, blue ones the number of transitions to absorption}
    \label{geoSplitPic}
    \end{subfigure}
\end{figure}

This construction extends to a general DPH with sub-transition matrix \textbf{T}. Denote the reward probabilities for state $i$ as $q_i$. The reward probabilities for each state are then given by the diagonal matrix $\textbf{Q}=diag(\textbf{q})$ with vector $\textbf{q}\in(0,1]^d$. Note that setting $q_i=1$ gives a degenerate geometric reward at zero, so non-rewarded states are included as a special case. The expanded sub-transition matrix is then

\begin{equation}\label{exMatGeo}
    B_g(\textbf{T},\textbf{Q})=\left(\begin{array}{cc}
         \textbf{Q}\textbf{T}&\textbf{I}-\textbf{Q}  \\
         \textbf{Q}\textbf{T}&\textbf{I}-\textbf{Q}
    \end{array}\right).
\end{equation}\label{foldTrmatGeo}

 This is again used in conjunction with the reward vector $\textbf{r}_1 = (\textbf{0}_d,\textbf{1}_d)^\top$ for the total reward and $\textbf{r}_2 = (\textbf{1}_d,\textbf{0}_d)^\top$ for the time.

The process cannot start in a reward-accumulating state, so the initial distribution becomes
 $\boldsymbol\beta=(\boldsymbol\pi,\textbf{0}_{d})^\top$. Importantly, the exit vector changes to $\textbf{b}=(\textbf{Qt},\textbf{Qt})$ (as the matrix blocks can be added up, leaving $(\textbf{I}-\textbf{QT}-(\textbf{I}-\textbf{Q}))\textbf{e}=\textbf{Q}(\textbf{I}-\textbf{T})\textbf{e}=\textbf{Qt}$). We therefore define the geometric random-reward DPH as follows.

\begin{definition}\label{geoRewDef}
    An observation $\textbf{Y}=(Y_1, Y_2)\in\mathbb N_0^2$ is said to come from a random reward DPH with geometric rewards, denoted $\textbf{Y}\sim \mathrm{RRDPH}_G(\boldsymbol\pi,\textbf{T},\textbf{P},
    (\textbf{r}_1,\textbf{r}_2))$ if $\textbf{Y}\sim MDPH(\boldsymbol\beta(\boldsymbol\pi),\textbf{B}_g(\textbf{T},\textbf{P}),(\textbf{r}_1,\textbf{r}_2))$, with $\boldsymbol\beta(\boldsymbol\pi)=(\boldsymbol\pi,\textbf{0})$ and $\textbf{B}_g(\textbf{T},\textbf{P})$ as in \eqref{foldTrmatGeo}.
\end{definition}

Informally speaking, the difference from the Bernoulli case is that time is halted while rewards are accumulated, as exemplified by the missing \textbf{T} on the right-hand side of the transition matrix. Thus, a transition between states of the original matrix \textbf{T} can only happen when rewards have stopped, e.g., when one geometric reward has concluded. The matrix order is also different: In the Bernoulli case, it is \textit{first} decided which state in \textbf{T} the transition occurs to, and whether it is rewarded comes \textit{second}. Meanwhile, the geometric RRDPH does \textit{first} check if the reward cycle has concluded, and if so, the \textit{second} step is a transition along the states in \textbf{T}.
Furthermore, this differs from the Bernoulli version, in which the unrewarded state is entered regardless. This makes it easier to separate time from rewards. 

Theorem \ref{geoRewTheo} states that $\mathrm{RRDPH}_G$ provides a valid distribution for the accumulated geometric rewards emitted by a DPH process. The proof can again be found in the appendix.

\begin{theorem}\label{geoRewTheo}
    Let $T\in[0,1]^{d\times d}$ be the sub-transition matrix of a DPH with initial probabilities $\pi\in\mathbb{R}^d$ and absorption vector $\textbf{t}=(\textbf{I}_d-\textbf{T})\textbf{e}$.  Let $\textbf{q}=(q_1,\dots,q_d)^\top$ and $\textbf{Q}=\text{diag}(q_1,\dots,q_d)$.  Consider a model where on each visit to state $i$ an independent $\mathrm{Geometric}(q_i)$ determines the reward of that visit, and let $X=(\psi,\tau)$ denote the accumulated rewards and transitions until absorption, respectively.  Let $Y=(Y_1,Y_2)\sim \mathrm{RRDPH}_G(\boldsymbol\pi,\textbf{T},\textbf{Q},(\textbf{r}_1,\textbf{r}_2))$
    with $\textbf{r}_1 = (\textbf{0}_d,\textbf{1}_d)^\top$
    and $\textbf{r}_2 = (\textbf{1}_d,\textbf{0}_d)^\top$.
     Then \textbf{X} and \textbf{Y} have the same distribution.
\end{theorem}
\addtocounter{theorem}{-2}

If one replaces random rewards with a fixed reward given by the mean of the reward distribution, this leads to quite a different stochastic behavior as shown in Figure \ref{distDiff}. 

\begin{figure}
    \centering
    \includegraphics[width=0.75\linewidth]{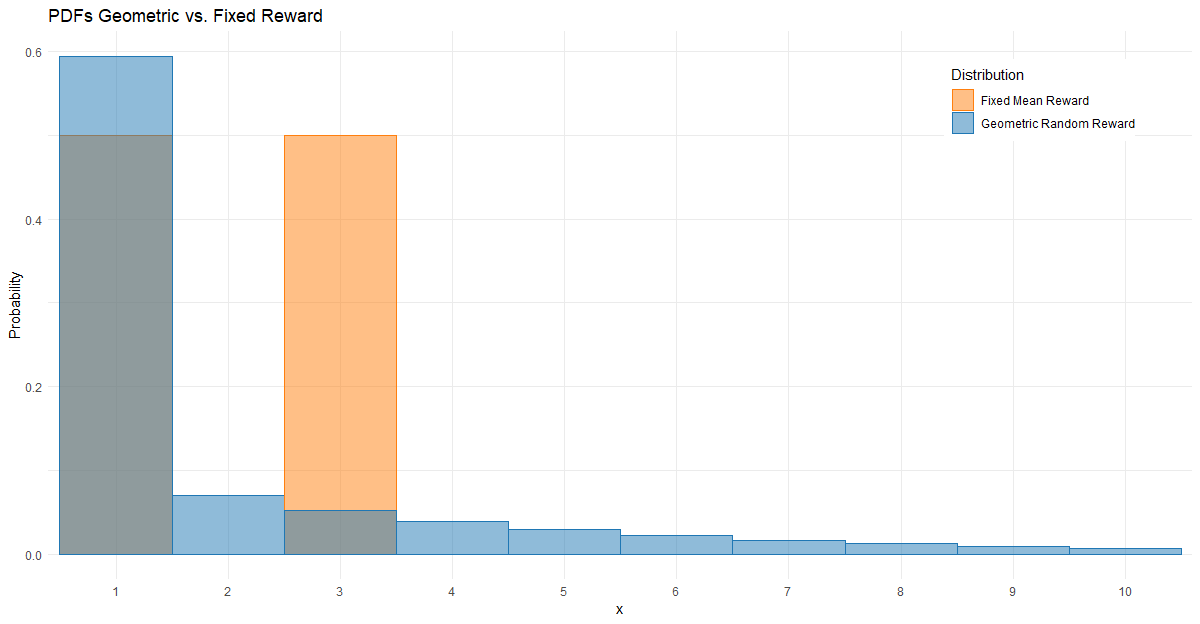}
    \caption{Random-Reward and Fixed-Reward model distributions for the geometric toy example as in Figure \ref{geoToyEx} and \ref{geoSplitPic}.
    The random reward PDF is $r(C)\sim Geom(\frac14)$. The fixed reward model assigns a reward of 3 to $C$, as 3 is the mean of a $Geom(\frac14)$ distribution. State B has  reward 1, A  0. }
    \label{distDiff}
\end{figure}

\section{Inertia-Escalation Model}\label{iemSec}

Imagine a latent process that can alter between $d\in\mathbb N$ different levels of severity. It is discrete in time, and with each time step, there are three possible transitions for a process with severity $i$: There might be \textit{de-escalation}: $i\rightarrow(i-1)$, \textit{inertia}: $i\rightarrow i$, or \textit{escalation}: $i\rightarrow (i+1)$. For a severity at time $t$, $X_t$, the probability of inertia (unchanging severity) is \[\nu=P(X_{t+1}=X_t),\]and, given that the severity changes, the probability of escalation is \[\eta=P(X_{t+1}=X_t+1\mid X_{t+1}\neq X_t).\]

The process finishes if de-escalation happens at severity 1 or escalation reaches severity d. For simplicity, both events will be treated as the same absorbing state here, meaning we do not distinguish between these two possible outcomes of the process. As each process will finish at some point, the absorption time $\tau$ can be modeled using a discrete Phase-Type (DPH) distribution with some starting vector $\boldsymbol\pi$ and sub-transition matrix

\begin{center}\begin{equation}\label{ieMod}
   \textbf{T}_{IE}=
        \bordermatrix{&\textbf{S}_1&\textbf{S}_2&\textbf{S}_3&...&\textbf{S}_{d-1}&S_d\cr        
         \textbf{S}_1&\nu&(1-\nu)\eta&0&...&0&0  \cr
         \textbf{S}_2&(1-\nu)(1-\eta)&\nu&(1-\nu)\eta&...&0&0\cr
         \textbf{S}_{...}&...&...&...&...&...&...\cr
         \textbf{S}_d&0&0&0&...&(1-\nu)(1-\eta)&\nu
    }.
\end{equation}\end{center}
\vspace{3mm}

One observable from the IEM will usually be the number of transitions $\tau\sim DPH(\boldsymbol\pi,\textbf{T}_{IE})$ until the absorbing state is reached (= the time the process took). We assume that at each transition into a state $X_t$, a Bernoulli/geometric reward $L_t\sim Ber(p_{X_t})/Geom(p_{X_t})$ is generated. After such a process has finished, the accumulation of these rewards, $\psi=\sum_{t=1}^\tau L_t$, can be observed. Multiple outcomes may also occur separately, and thus, the IEM might be expanded to a multivariate setup.

As an example of such a process, consider conflict analysis: The IEM would model the severity of the conflict, and $\tau$ would be the number of weeks until the conflict terminates. Additionally, one could observe the total number of weeks during which a ceasefire or a major battle occurred (both with Bernoulli rewards), as well as the total number of military and civilian deaths, and the total amount of ammunition expended. In the latter cases, it needs to be checked whether geometric rewards fit the data.

Here, we define the IEM to model only with geometric rewards. As such, references to IEM in the rest of the paper will adhere to the definition \ref{iemDef}.
We assume that $n$ observation pairs $Y_{i,}=(\tau_i,\psi_i)$ are available, with $1\le i\le n$ that 
consists of the times to absorption and the
accumulated rewards.

\begin{definition}\label{iemDef}
    An observation matrix $\textbf{Y}\in\mathbb R_n^2$ is said to have been generated by an Inertia Escalation Model (IEM) with parameters $\nu$,$\eta$, and \textbf{Q}, denoted \[\textbf{Y}\sim IEM(\nu,\eta,\textbf{Q}),\]  if for observation $i$ $\textbf{Y}_{i,}\sim \mathrm{RRDPH}_G(\textbf{e}_1,\textbf{T},\textbf{Q})$, where \textbf{T} is constructed as in \eqref{ieMod} using $\nu$ and $\eta$.
\end{definition}

\subsection{Inertia Escalation Model with Covariates}\label{regIEMsec}
In many applications, repeated observations from the same underlying Markov chain are unavailable, or the parameters governing the chain are expected to vary across subjects. To accommodate this, we extend the Inertia Escalation Model (IEM) by allowing the parameters $\nu$ and $\eta$ to depend on covariates. Since both parameters will lie in $(0,1)$, a logistic link is a natural choice.

Let $Y_i$ denote the observation vector for subject $i$, and let $X_i$ be the associated covariate vector. We write $\beta_\nu$ and $\beta_\eta$ for the corresponding regression coefficient vectors and propose the regression-adapted RRDPH in definition \ref{regIemDef},
where $\\\logit(x) = \log\!\left(\frac{x}{1-x}\right)$.

\begin{definition}[Regression IEM]\label{regIemDef}
An observation vector $Y \in \mathbb{R}_n^2$ is said to follow a regression IEM, denoted
\[
Y \sim \mathrm{IEM_R}(\beta_\nu,\beta_\eta,Q,X),
\]
if, for each observation $i$,
\[
\logit(\nu_i) = X_i^\top \beta_\nu,
\qquad
\logit(\eta_i) = X_i^\top \beta_\eta,
\]
and
\[
Y_i \sim \mathrm{RRDPH}_G(e_1, T_i, Q),
\]
where $T_i$ is the sub-transition matrix constructed as in \eqref{ieMod} using $\nu_i$ and $\eta_i$.
\end{definition}

Thus, the quantities that are estimated in the regression IEM are the regression coefficients $\beta_\nu$ and $\beta_\eta$, rather than a single common pair $(\nu,\eta)$. In the simplest case of a single numerical covariate, this replaces estimating two scalar parameters with estimating four coefficients (two intercepts and two slopes). More generally, if $r$ covariates are used, then the model contains $2(r+1)$ regression coefficients.

The reward probabilities may also be modeled using logistic regression. Instead of estimating the diagonal entries of $Q = \mathrm{diag}(q_1,\dots,q_d)$ freely, one may, for instance, impose a linear reward model across stages,
\[
\logit(q_i) = \beta_0^q + \beta_1^q\, i,
\qquad i=1,\dots,d.
\]
This reduces the reward component to two parameters and imposes a monotone escalation structure across the latent states, which is often easier to interpret.

\section{Statistical Inference}\label{inferenceSec}

Since the IEM is a special case of the MDPH distribution, estimation procedures can be obtained by specializing the EM algorithm of He \& Ren \cite{He2016Sep}.  The appendix provides the E-step and closed-form M-step updates; here, we state only the resulting estimation strategy. Note further that other versions of RRDPH that are not the IEM can use the EM algorithm analogously, only updating the maximization functions.

Throughout this section, we write
\[
Y_i \sim \mathrm{IEM}(\nu,\eta,\mathbf{Q}), \qquad i=1,\dots,n,
\]
where \(Y_i=(\tau_i,\psi_i)\), \(\tau_i\) denotes the absorption time, and \(\psi_i\) denotes the accumulated geometric reward.

Let \(X_i=(X_{i,0},X_{i,1},\dots,X_{i,\tau_i})\) denote the unobserved latent state sequence associated with observation \(Y_i\). 
The complete data are therefore given by the latent transitions and the latent reward counts induced by the geometric emissions. For \(1\leq j,k\leq d\), define the transition counts
\[
N^{(i)}_{jk}=\sum_{t=1}^{\tau_i}\mathbf{1}\{X_{i,t-1}=j,\;X_{i,t}=k\},
\]
and let \(N^{(i)}_{j0}\) denote the number of transitions from state \(j\) to absorption.

Given current parameter values \(\theta^{(m)}=(\nu^{(m)},\eta^{(m)},Q^{(m)})\), the E-step computes the conditional expectations
\[
\widehat N^{(i)}_{jk}
=
\mathbb{E}\!\left[N^{(i)}_{jk}\mid Y_i,\theta^{(m)}\right],
\qquad
\widehat N^{(i)}_{j0}
=
\mathbb{E}\!\left[N^{(i)}_{j0}\mid Y_i,\theta^{(m)}\right],
\]
together with the corresponding expected reward counts. These quantities are obtained from the forward recursion described in Appendix 10.3.

The expected complete-data log-likelihood can be separated into a transition part and a reward part, making the M-step straightforward. The updated parameter vector is obtained by
\[
\theta^{(m+1)}=\arg\max_{\theta}\;
\mathbb{E}\!\left[\ell_c(\theta)\mid Y_1,\dots,Y_n,\theta^{(m)}\right].
\]
Because the IEM transition matrix has only inertia and nearest-neighbor escalation/de-escalation moves, the updates for \(\nu\) and \(\eta\) are available in closed form. The geometric reward probabilities \(q_1,\dots,q_d\) are updated by the corresponding closed-form ratio of expected reward failures to expected opportunities, as derived in Appendix 10.3.

If some parameters are known in advance, the same EM framework applies: the fixed parameters are simply held constant during maximization, and only the unknown parameters are updated.

\paragraph{Regression IEM.} In the regression IEM introduced in section \ref{regIEMsec}, subject-specific inertia and escalation probabilities are modeled via logistic links. Given the E-step pseudo-counts, the M-step updates for $\beta_\nu$ and $\beta_\eta$
reduce to weighted quasibinomial GLMs. If the reward probabilities are modeled through the linear reward specification, the same weighted-GLM update applies to the reward coefficients.

\section{Simulation}\label{simSec}

To assess the finite-sample behavior of the proposed estimation procedures, we conducted simulation studies for the random-reward models introduced above. The first experiments revisit the Bernoulli and geometric toy examples from section \ref{rrSec}, where the goal is to recover all parameters under correctly specified data-generating mechanisms. These examples serve as a simple check that the random-reward construction and the EM algorithm work as intended in low-dimensional settings.

We then focus on the Inertia-Escalation Model. In particular, we study two variants of the regression-adapted IEM from Definition \ref{regIemDef}: one in which the reward probabilities are modeled by the linear reward specification introduced in Section \ref{regIEMsec}, and one in which the reward probabilities are estimated freely. This allows us to assess both the transition-parameter estimates and the reward-parameter estimates under the inference framework developed in Section \ref{inferenceSec}.

For the reward probabilities in the IEM, we used an underlying true model with an intercept of -3.064788 and a slope coefficient of 0.8675632, so that the reward probabilities for the first and fourth states are 0.1 and 0.6, respectively. The regression coefficients for $\nu$ and $\eta$ are given in Table \ref{iemCoefTab}. The covariates were sampled with replacement from a vector of (-10, 0, 5, 20). 

\begin{table}[ht]
\centering
\begin{tabular}{r|rr}
  \hline
 & $\beta_0$ & $\beta_1$  \\ 
  \hline
$\nu$ & -0.1 & 0.2  \\ \hline
$\eta$ & 0.1 & -0.25 \\ 
   \hline
\end{tabular}\caption{Regression Coefficients for simulations on $IEM_R$}\label{iemCoefTab}
\end{table}

In all simulation settings, data were generated from the corresponding model, and the proposed EM procedure was used for estimation. The subsections below report the results for each model separately. Each scenario consists of 200 replicates, in which 1000 observations were simulated each.

\subsection{Toy Example - Bernoulli Rewards }

To simulate from the Bernoulli toy example (illustrated in Figure \ref{berToyExNEW}), the following model was used:
\begin{equation}
    Y\sim \mathrm{RRDPH}_B\left((1,0,0,0),\left(\begin{array}{cccc}
     0&b&1-b&0  \\
     0&0&0&1\\
     0&0&0&0\\
     0&0&0&0\\
\end{array}\right),\left(\begin{array}{cccc}
     1&0&0&0  \\
     0&1&0&0\\
     0&0&p&0\\
     0&0&0&q
\end{array}\right)\right),
\end{equation}

with $b:=\frac12, q:=0.3$ and $p:=0.6$. These three parameters, b, q, and p, were estimated using 1000 draws from this distribution. This procedure was replicated 200 times to assess the volatility of parameter estimates. The estimates seem to be unbiased and have relatively low variance, as shown in Figure \ref{berErr}. 

\begin{figure}
\begin{minipage}{.45\textwidth}
    \centering
    \includegraphics[width=\linewidth]{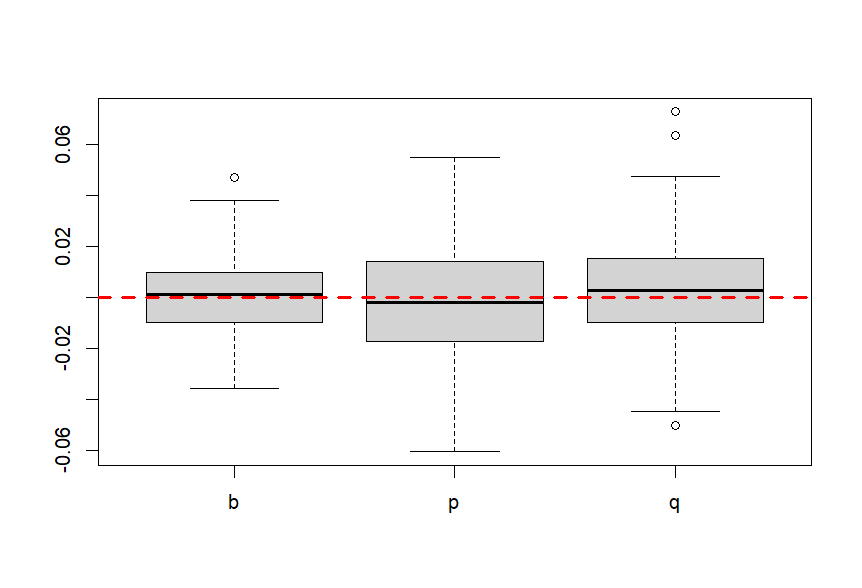}
    \caption{Errors of estimations of b,p, and q for the Bernoulli toy example using 200 replicates}
    \label{berErr}
\end{minipage}
\begin{minipage}{.45\textwidth}
    \centering
    \includegraphics[width=\linewidth]{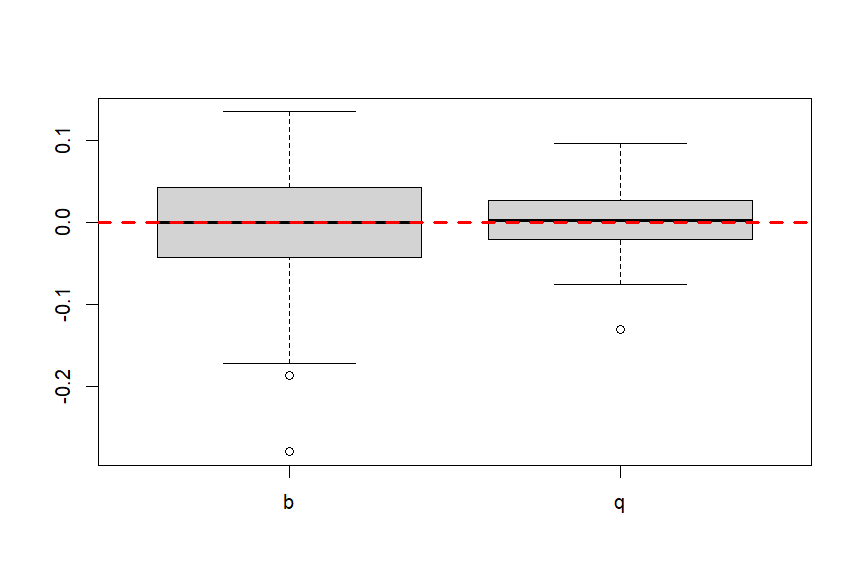}
    \caption{Errors of estimations of b and q for the geometric toy example using 200 replicates}
    \label{geoErr}
\end{minipage}
\end{figure}

\subsection{Toy Example - Geometric Rewards }

The toy example of geometric rewards, illustrated in Figure \ref{geoSplitPic}, is more computationally intensive because larger rewards may be accumulated per draw. Nevertheless, we again estimate parameters from 1000 draws and replicate this 200 times. The model is specified to be

\begin{equation}
      Y\sim \mathrm{RRDPH}_G\left((1,0,0),\left(\begin{array}{ccc}
     0&b&1-b  \\
     0&0&0\\
     0&0&0\\
\end{array}\right),\left(\begin{array}{ccc}
     1&0&0  \\
     0&1&0\\
     0&0&q
\end{array}\right)\right) . 
\end{equation}

The results again support unbiased estimation, albeit with higher variance, as shown in Figure \ref{geoErr}.

\subsection{IEM with Reward-Regression Model}

Here, we simulated $Y\sim IEM_R(\boldsymbol\beta_\nu,\boldsymbol\beta_\eta,\textbf{Q},\textbf{X})$, with $\boldsymbol\beta_\nu$ and $\boldsymbol\beta_\eta$ given in Table \ref{iemCoefTab} and estimated the parameters using the regression adaptation for \textbf{Q}.

Figure \ref{estModelPlot} shows the results of simulations using a model of the reward probabilities. The estimation for $\beta_0^\nu$ is a bit off, but not significantly different from the true value.  The reward probabilities are captured nicely, meaning that the two parameters, the intercept and slope, were estimated correctly. 
\begin{figure}
\begin{minipage}{0.45\textwidth}
    \centering
    \includegraphics[width=\linewidth]{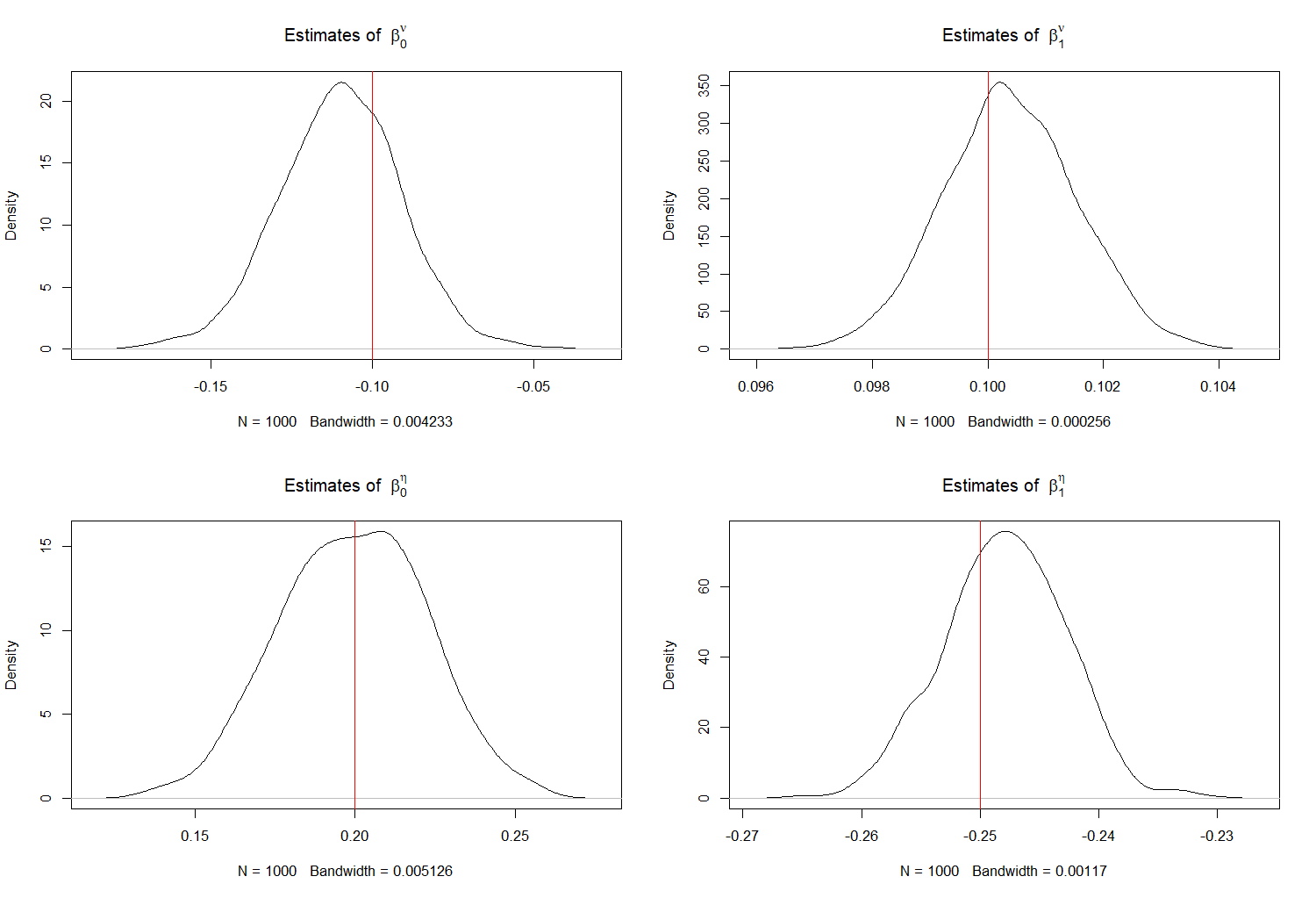}
    \end{minipage}\hfill
\begin{minipage}{0.45\textwidth}
    
    \centering
    \includegraphics[width=\linewidth]{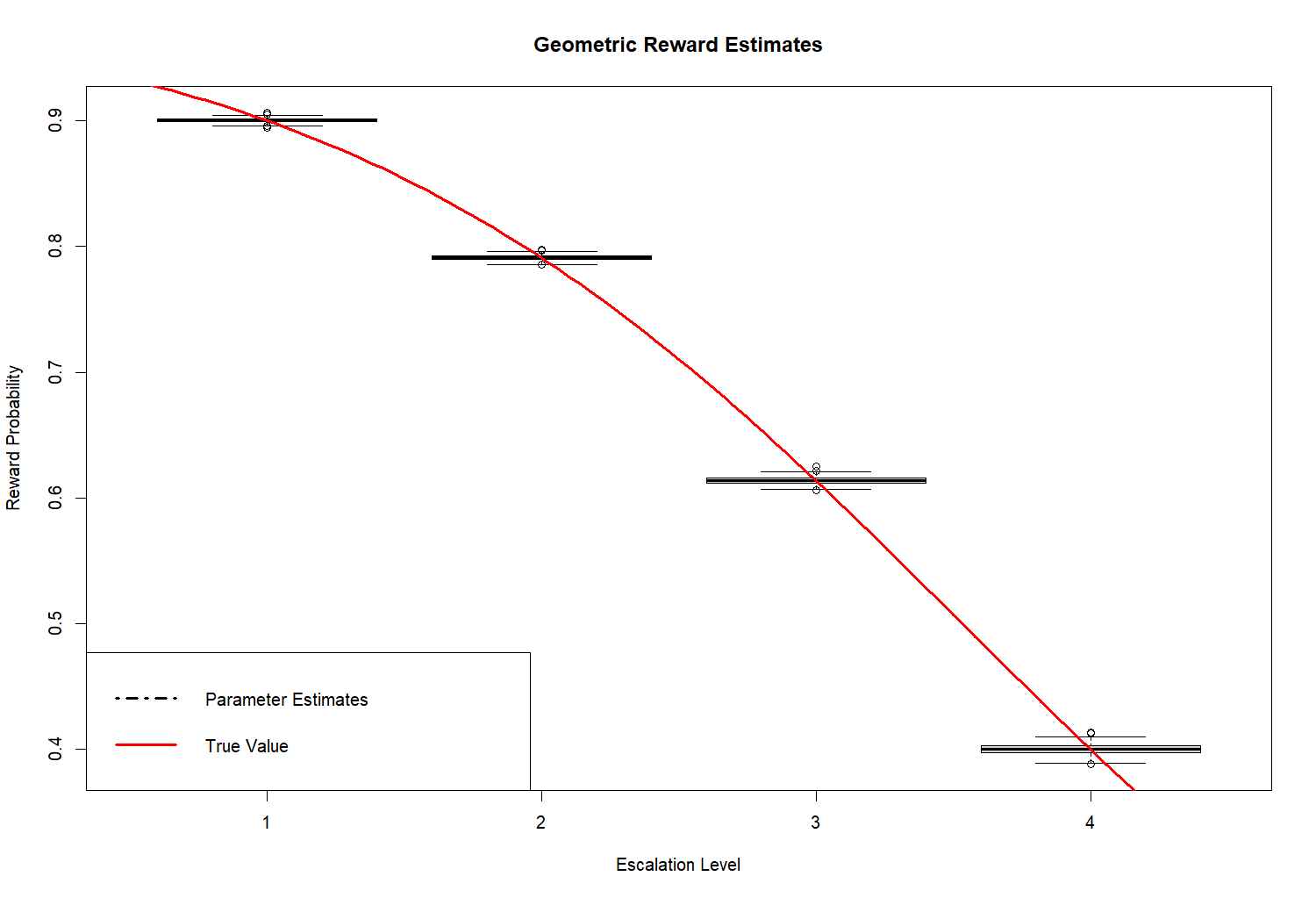}
\end{minipage}
    \caption{Estimation results, coefficients, and reward probabilities with regression model. Kernel density estimates for the coefficients are displayed (left), and boxplots for the reward probabilities (right). The true values and the underlying true regression curve are red.}
    \label{estModelPlot}
\end{figure}

\subsection{IEM with Free Reward Estimation}

Here, we simulated $Y\sim IEM_R(\boldsymbol\beta_\nu,\boldsymbol\beta_\eta,\textbf{Q},\textbf{X})$, with $\boldsymbol\beta_\nu$ and $\boldsymbol\beta_\eta$ given in Table \ref{iemCoefTab}, and estimated the parameters of \textbf{Q} freely.
Figure \ref{estFreePlot} shows promising results even for the free parameter estimation. Note that the true, underlying model (red) is still the same as in the previous section. This time, however, we do not exploit that structure. 

Still, the results for the reward probabilities seem unbiased and have a low variance. However, the kernel density of the parameter estimates across all replications appears more skewed this time, possibly due to the parameters' limited support. Once again, $\beta_0^\nu$ seems to be the most difficult to estimate. 
\begin{figure}
\begin{minipage}{0.45\textwidth}
    \centering
    \includegraphics[width=\linewidth]{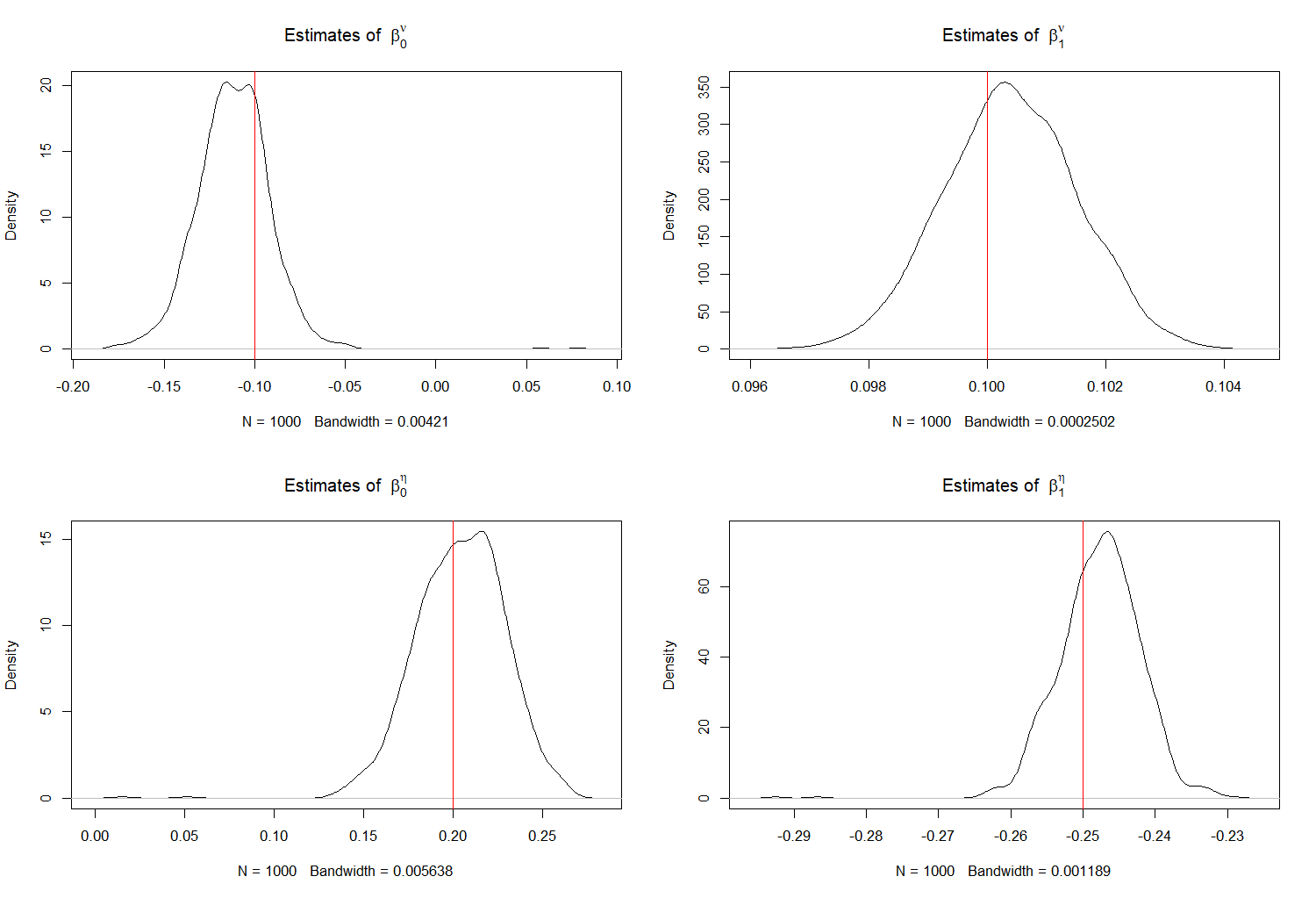}
\end{minipage}\hfill
\begin{minipage}{0.45\textwidth}
    \centering
    \includegraphics[width=\linewidth]{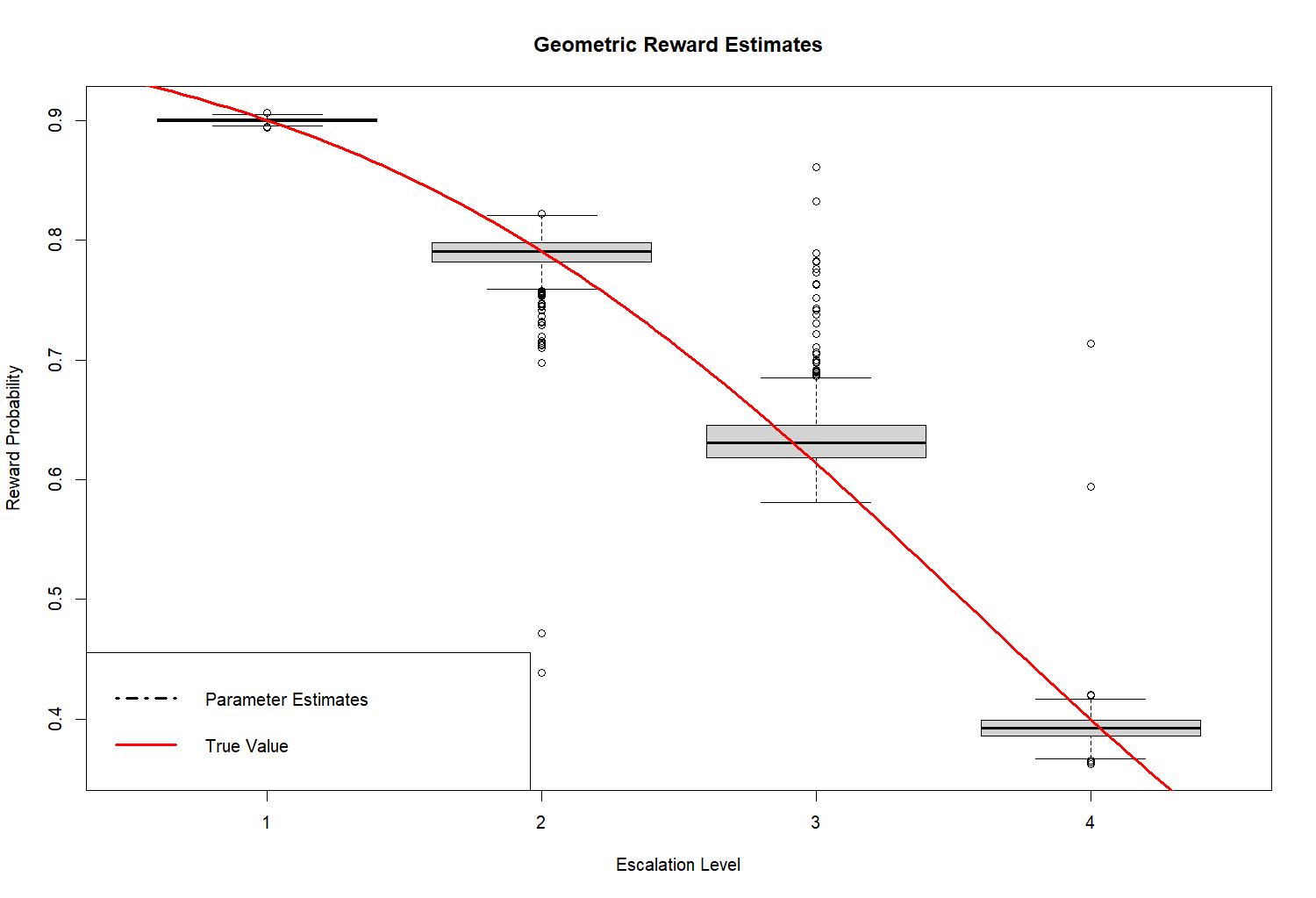}
    \end{minipage}
    \caption{Estimation results for regression coefficients, and freely estimable reward probabilities. The estimations are portrayed as kernel density estimates for the coefficients (left) and boxplots for the reward probabilities (right). The true values and the underlying true regression curve are red.}
    \label{estFreePlot}
\end{figure}

\section{Application}\label{ApplSec}

We illustrate the proposed IEM with random rewards on two datasets. The first is the Telco Customer Churn data \cite{telcoData}, which records customer tenure and billing information for a fictional telecommunication company. The second is the Brecke conflict catalog \cite{histData}, which contains historical conflict data. In both applications, the latent IEM models severity over time, while geometric rewards represent accumulated quantities such as total charges or total casualties. The geometric specification is not intended to be a literal mechanism in either case; rather, it provides a flexible and tractable way to model aggregated counts and totals.

\subsection{Telco Customer Churn}

Three variables from the Telco data are considered in our analysis: Contract, Total Charges, and Tenure. To keep things simple, $contract$ was simplified to a binary variable that could be either "month-to-month" or a longer-form contract. The month-to-month contract is the reference class, so the $\beta_1$ will show the difference that a longer-form contract makes. The tenure is defined as the time to absorption, meaning the time until a customer churns or leaves the company. Total charges are the total amount of money that the customer has paid during their tenure. This can be modeled naturally using random rewards. Total charges are treated as accumulated rewards and are therefore modeled using the geometric reward component of the IEM.
\begin{figure}
    \centering
    \includegraphics[width=0.5\linewidth]{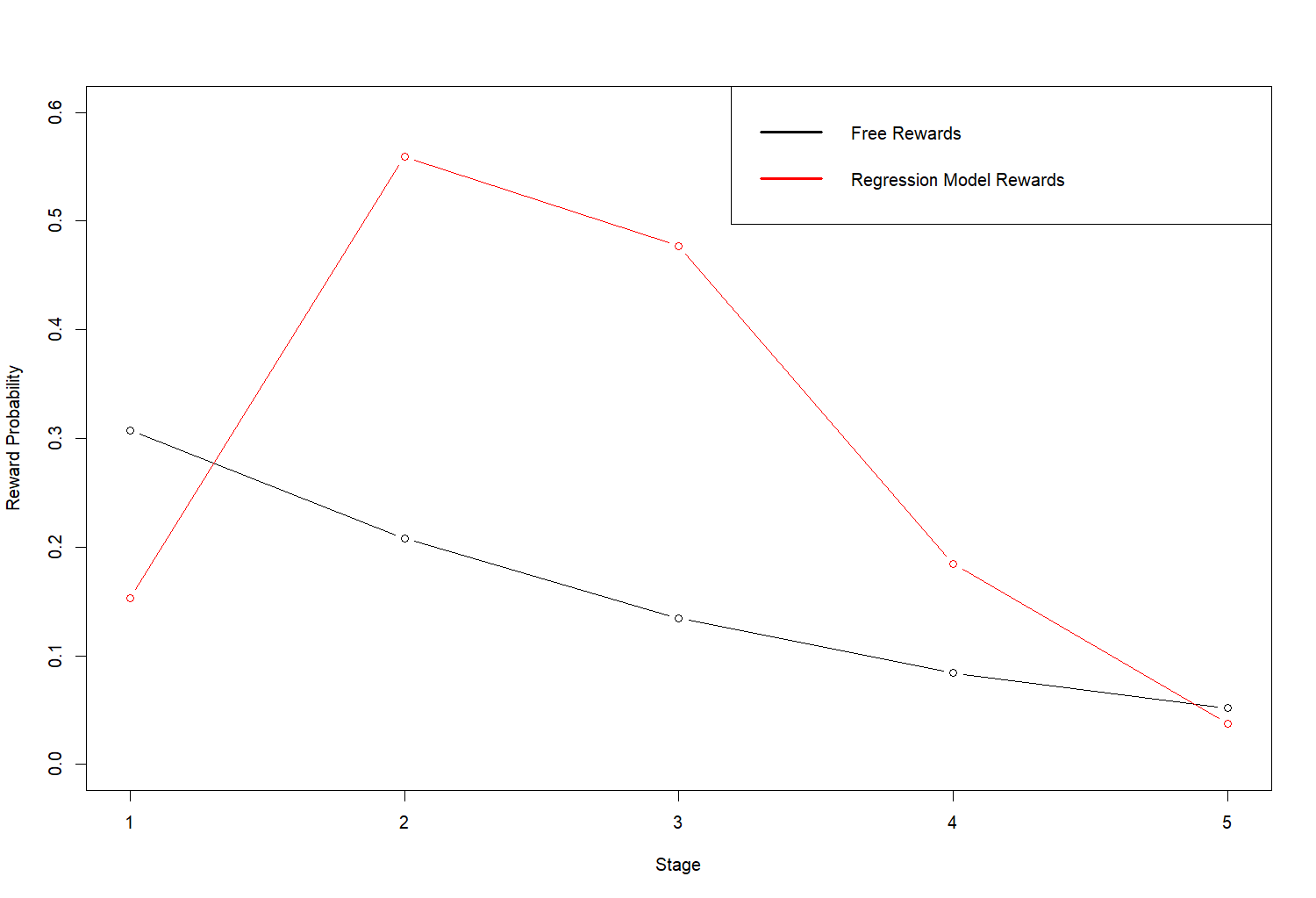}
    \caption{Estimated reward probabilities for the different stages for the Telco Customer Churn data. Estimates from the regression model are in red, freely varying estimates in black.}
    \label{telcRewPlot}
\end{figure}

Figure \ref{telcRewPlot} compares the estimated reward probabilities obtained under the regression reward model and under free stage-specific reward estimation. The two fits differ substantially. In particular, the freely estimated reward probabilities exhibit a clearly non-linear pattern across the latent states. This suggests that a linear reward specification is not suitable for this dataset. The fitted curves indicate that the charging process changes across latent severity levels. This might be interpreted to mean that the total charges at the beginning are quite high, and from there on, either one stops being a customer or reduces the payments significantly, until eventually, one reaches a higher level of customer severity again, where more subscriptions are added, and the payment is higher again in these phases.

Table \ref{telcCoefTable} shows the estimated coefficients in both settings. They are somewhat different, which can be attributed to the non-linear behavior of the freely estimated rewards that were seen in Figure \ref{telcRewPlot}. Notably, all coefficients are larger than zero, meaning that we have, in any case, $\eta>\frac12$ and $\nu>\frac12$. The positive slope estimates indicate that a longer-form contract is associated with higher latent severity levels and with longer expected tenure. This is consistent with the interpretation that customers on longer contracts tend to remain in the system longer, although the estimated effect should be read as an association rather than a causal statement.

\begin{table}[ht]\caption{Estimated Regression Coefficients}\label{telcCoefTable}
\centering
\begin{tabular}{r|rrrr}
  \hline
 & $\beta_0^\nu$ & $\beta_1^\nu$ & $\beta_0^\eta$ & $\beta_1^\eta$ \\ 
  \hline
Regression Rewards & 0.40 & 1.85 & 0.63 & 0.81 \\ \hline
Free Rewards & 0.78 & 1.35 & 0.79 & 1.82 \\ 
   \hline
\end{tabular}
\end{table}

\subsection{Casualties in Historical Warfare}

Historical warfare settings are also a suitable application for models such as the IEM. This is particularly the case if month-by-month data for events long past are not recorded in detail, but the total amount of certain measurements is available. As such, we consider three variables from the historical conflict dataset: Total duration, total military fatalities, and number of participating actors. The duration is calculated in monthly units, so if the number of days is available, it is divided by 30. As there are some clear, recent outliers, we opted to consider only conflicts before 1900, which excludes both world wars and makes the remaining conflicts somewhat more homogeneous.

As in the Telco application, the geometric reward component is used to model the accumulated total of a measured quantity, here military fatalities. Figure \ref{histRewPlot} shows the estimated reward probabilities under the regression reward model and under free stage-specific reward estimation. In this application, the two fits are much closer than in the Telco case, which suggests that the linear reward specification is adequate here. The estimated reward probabilities increase with the latent state index. Since larger geometric success probabilities correspond to smaller expected counts, this pattern indicates that wars start with high casualties, after which they either end quickly or casualties begin to decrease.

\begin{figure}
    \centering
    \includegraphics[width=0.5\linewidth]{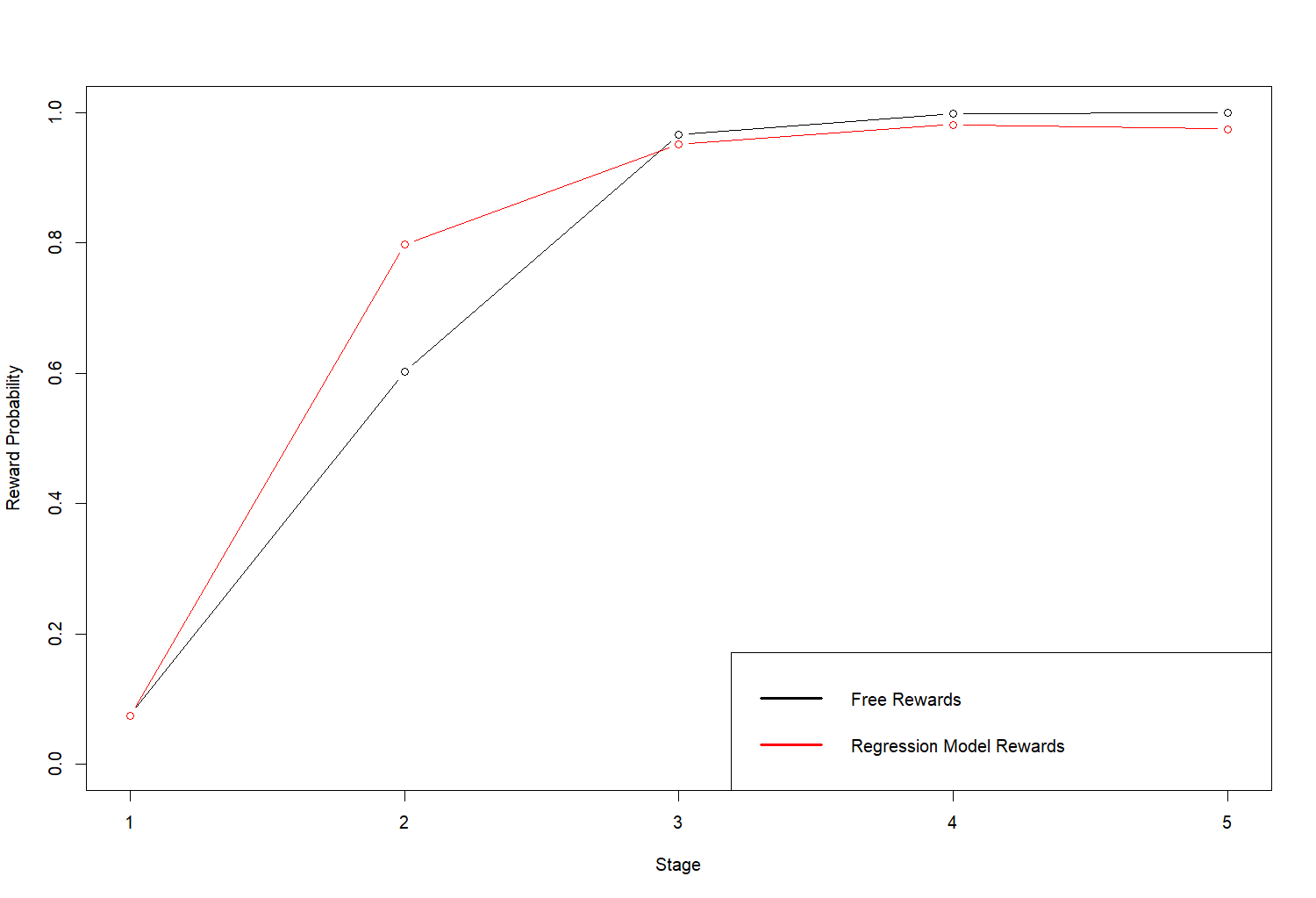}
    \caption{Estimated reward probabilities for the different stages for the historical conflict data. Estimates from the regression model are in red, freely varying estimates in black.}
    \label{histRewPlot}
\end{figure}

Table \ref{histCoefTable} gives the corresponding coefficient estimates. The positive slope estimates indicate that conflicts with more participants tend to last longer, and that longer conflicts are associated with progression toward lower-casualty states. This interpretation aligns with the fitted reward pattern and is broadly consistent with the idea that long-lasting conflicts often become less intense over time.

\begin{table}[ht]\caption{Estimated Regression Coefficients}\label{histCoefTable}
\centering
\begin{tabular}{r|rrrr}
  \hline
 & $\beta_0^\nu$ & $\beta_1^\nu$ & $\beta_0^\eta$ & $\beta_1^\eta$ \\ 
  \hline
Regression Rewards & -0.03 & 0.31 & 0.02 & 0.20 \\ \hline
Free Rewards & -0.07 & 0.31 & 0.01 & 0.20 \\ 
   \hline
\end{tabular}
\end{table}

\section{Discussion}\label{discSec}
The random rewards proposed in this paper go an important step in increasing the flexibility of modeling with DPH, but research in this area needs to continue. The two distributional families developed here, Bernoulli and geometric, have some applications. However, many situations in which random rewards might be used would require different, more suitable distributional families. 

Furthermore, a comparison of the IEM with alternative modeling strategies could be interesting. It also remains to be seen which further insights can be generated through the IEM. An example can be seen in our application section. There, the sign of the reward probability slope coefficient indicates whether rewards tend to accumulate at the beginning or end of a trajectory. 

Finally, the applicability of random rewards is not limited to the IEM; there are further situations where it might be useful. For example, in a missing-data process where not all transitions are observed, Bernoulli rewards might provide insight into the process if the observation probability is constant throughout the chain.

\backmatter

\bmhead{Acknowledgements}
We would like to extend our thanks to Qi-Ming He, who shared the code of his original paper on the EM algorithm for MDPH \cite{He2016Sep} with us. Furthermore, we thank IES 2025 for the opportunity to present this research at their conference.
\section{Data Availability Statement}
No new data were produced for this paper. The two datasets used in the application section are both publicly available. The historical warfare data \cite{histData} can be found at https://datahub.io/collections/war-and-peace, under the section "Brecke Conflict Catalog". The Telco Customer Churn data \cite{telcoData} can be found here:

\hspace{-7.5mm}
https://www.kaggle.com/datasets/blastchar/telco-customer-churn

\section*{Declarations}

Not applicable.

\begin{appendices}

\section{Proof of Theorem 1}

We restate the theorem and prove it:
\begin{theorem}\notag
Let $T\in[0,1]^{d\times d}$ be the sub-transition matrix of a DPH with initial probabilities $\pi\in\mathbb{R}^d$ and absorption vector $\textbf{t}=(\textbf{I}_d-\textbf{T})\textbf{e}$.  Let $\textbf{p}=(p_1,\dots,p_d)^\top$ and $\textbf{P}=\text{diag}(p_1,\dots,p_d)$.  Consider a model where on each visit to state $i$ an independent $\mathrm{Bernoulli}(p_i)$ determines whether the visit is rewarded, and let $X=(\psi,\tau-\psi)$ denote the rewarded and unrewarded counts, respectively (so $\psi$ is the number of rewarded visits and $\tau$ the absorption time). Let $Y=(Y_1,Y_2)\sim \mathrm{RRDPH}_B(\boldsymbol\pi,\textbf{T},\textbf{P})$ as in definition \ref{berRewDef}.  Then \textbf{X} and \textbf{Y} are distributionally equivalent.
\end{theorem}

\begin{proof}
We give a direct algebraic proof by computing both PGFs and showing that they are equal. First, we consider \textbf{Y}.
By the matrix PGF formula for MDPH (see \ Eq.\ (5.6) in \cite{navarro}) the PGF of $Y$ is
\[
G_Y(\theta_1,\theta_2)\;=\;\mathbb E(\theta_1^{Y_1}\theta_2^{Y_2})=
\boldsymbol\beta^\top \Delta \underbrace{\left(\textbf{I}_{2d} - \textbf{B}_b\Delta\right)^{-1} \textbf{b}}_{:=\textbf{u}},
\]
where $\textbf{b}=(\textbf{t},\textbf{t})$, $\textbf{B}_b$ and $\boldsymbol\beta$ follow from definition \ref{berRewDef}, and
\[
\Delta=\text{diag}(\underbrace{\theta_2,\dots,\theta_2}_{d},
\underbrace{\theta_1,\dots,\theta_1}_{d}).
\]
Write the block factors
\[
\textbf{A}:=\textbf{T}(\textbf{I}-\textbf{P}),\qquad \textbf{B}:=\textbf{TP}\qquad(\text{so }\textbf{A}+\textbf{B}=\textbf{T}).
\]
Then
\[
\textbf{B}_b\Delta
=
\begin{pmatrix} \textbf{A}\theta_2 & \textbf{B}\theta_1 \\[4pt] \textbf{A}\theta_2 & \textbf{B}\theta_1 \end{pmatrix},
\qquad
\textbf{I}_{2d}-\textbf{B}_b\Delta
=
\begin{pmatrix} \textbf{I}_d - \textbf{A}\theta_2 & -\textbf{B}\theta_1 \\[4pt] -\textbf{A}\theta_2 & \textbf{I}_d - \textbf{B}\theta_1 \end{pmatrix}.
\]
Using $\textbf{u}:=(\textbf{I}_{2d}-\textbf{B}_b\Delta)^{-1}\textbf{b}=\begin{pmatrix}\textbf{u}_1\\[4pt]\textbf{u}_2\end{pmatrix}$ (with $\textbf{u}_1,\textbf{u}_2\in\mathbb{R}^d$), the system
\[
(\textbf{I}_{2d}-\textbf{B}_b\Delta)\textbf{u}=\textbf{b}
\]
can be split up into the equations
\begin{align}
(\textbf{I}_d - \textbf{A}\theta_2)\textbf{u}_1 - \textbf{B}\theta_1\,\textbf{u}_2 &= \textbf{t}, \label{eq:block1}\\[4pt]
- \textbf{A}\theta_2\,\textbf{u}_1 + (\textbf{I}_d - \textbf{B}\theta_1)\textbf{u}_2 &= \textbf{t}. \label{eq:block2}
\end{align}
Subtracting \eqref{eq:block2} from \eqref{eq:block1} eliminates the $\textbf{A}\theta_2$ and $\textbf{B}\theta_1$ terms and yields
\[
\textbf{u}_1 - \textbf{u}_2 = \textbf{0} \quad\Longrightarrow\quad \textbf{u}_1=\textbf{u}_2=:\textbf{u}_d.
\]
Adding \eqref{eq:block1} and \eqref{eq:block2} and using $\textbf{A}\theta_2+\textbf{B}\theta_1 = \textbf{T}\left((\textbf{I}-\textbf{P})\theta_2+\textbf{P}\theta_1\right)$ gives
\[
\left(\textbf{I}_d - \textbf{TD}\right)(\textbf{u}_1+\textbf{u}_2) = 2\,\textbf{t},\text{ where }\]\[
\qquad \textbf{D} := \text{diag}\left((1-p_i)\theta_2+p_i\theta_1\right)_{i=1}^d=(\textbf{I}-\textbf{P})\theta_2+\textbf{P}\theta_1.
\]
Since $\textbf{u}_d=\textbf{u}_1=\textbf{u}_2$ we obtain
\[
\textbf{u}_d = (\textbf{I}_d - \textbf{TD})^{-1} \textbf{t}.
\]
Now restate $G_Y$ using lower-dimensional matrices:
\begin{align*}
G_Y(\theta_1,\theta_2)
&= \boldsymbol\beta^\top \Delta \textbf{u}
= \left(((\textbf{I}-\textbf{P})\pi)^\top\theta_2,\;(\textbf{P}\boldsymbol\pi)^\top\theta_1\right)
\begin{pmatrix}\textbf{u}_1\\[4pt]\textbf{u}_2\end{pmatrix}\\[4pt]
&= \boldsymbol\pi^\top\left( (\textbf{I}-\textbf{P})\theta_2\,\textbf{u}_1 + \textbf{P}\theta_1\,\textbf{u}_2\right)
= \boldsymbol\pi^\top\left((\textbf{I}-\textbf{P})\theta_2 + \textbf{P}\theta_1\right) \textbf{u}_d\\[4pt]
&= \boldsymbol\pi^\top \textbf{D}\,(\textbf{I}_d - \textbf{TD})^{-1} \textbf{t}.
\end{align*}

Now we consider the PGF of \textbf{X}. The PGF of a Bernoulli($p_i$) success/failure tuple (s,f) is \[\sum\limits_{(s,f)\in\{(1,0),(0,1)\}}p_i^s(1-p_i)^f\theta_1^s\theta_2^f=(1-p_i)\theta_2+p_i\theta_1.\]
Collect the PGFs of these Bernoulli distributions in a diagonal matrix \[\textbf{D}=\text{diag}((1-p_i)\theta_2+p_i\theta_1)_{i=1}^d=(\textbf{I}-\textbf{P})\theta_2+\textbf{P}\theta_1.\] Then the pathwise contribution to the PGF is
\[G_X(\theta_1,\theta_2) = \boldsymbol\pi\textbf{Dt}+\boldsymbol\pi\textbf{DTDt}+\boldsymbol\pi\textbf{D}(\textbf{TD})^2\textbf{t}+\ldots=\sum\limits_{n\ge0}\boldsymbol\pi\textbf{D}(\textbf{TD})^n\textbf{t}.\]

For $(\theta_1,\theta_2)\in[0,1]^2$ we have $\|\textbf{TD}\|\le\|\textbf{T}\|<1$ (since $\textbf{T}$ is sub-stochastic and $\max_i d_i\le1$), so the Neumann series converges and
\[
G_X(\theta_1,\theta_2)=\boldsymbol\pi^\top \textbf{D}(\textbf{I}_d - \textbf{TD})^{-1} \textbf{t}.
\]

Comparing the expressions derived for $G_Y$ and $G_X$ shows
\[
G_Y(\theta_1,\theta_2)=\boldsymbol\pi^\top \textbf{D}(\textbf{I}_d - \textbf{TD})^{-1} \textbf{t}=G_X(\theta_1,\theta_2).
\]
Hence, the two joint PGFs coincide. Therefore, \textbf{X} and \textbf{Y} have equivalent distributions.

\end{proof}

\section{Proof of Theorem 2}

We restate the theorem and prove it:

\begin{theorem}
    Let $T\in[0,1]^{d\times d}$ be the sub-transition matrix of a DPH with initial probabilities $\pi\in\mathbb{R}^d$ and absorption vector $\textbf{t}=(\textbf{I}_d-\textbf{T})\textbf{e}$.  Let $\textbf{q}=(q_1,\dots,q_d)^\top$ and $\textbf{Q}=\text{diag}(q_1,\dots,q_d)$.  Consider a model where on each visit to state $i$ an independent $\mathrm{Geometric}(q_i)$ determines the reward of that visit, and let $X=(\psi,\tau)$ denote the accumulated rewards and transitions until absorption, respectively.  Let $Y=(Y_1,Y_2)\sim \mathrm{RRDPH}_G(\boldsymbol\pi,\textbf{T},\textbf{Q})$.  Then \textbf{X} and \textbf{Y} are distributionally equivalent.
\end{theorem}
\begin{proof}
    This proof is similar to that of theorem \ref{berRewTheo}, but the cyclical structure of geometric rewards increases the complexity. We again show that the joint PGFs coincide. For \textbf{Y}, the PGF is \[G_Y(\theta_1,\theta_2)=\boldsymbol\beta\Delta(\textbf{I}_{2d}-\textbf{B}_g\Delta)^{-1}\textbf{b},\] where $\textbf{b}=(\textbf{Qt},\textbf{Qt})$, $\textbf{B}_g$ and $\boldsymbol\beta$ follow from definition \ref{geoRewDef}, and \[\Delta=\text{diag}(\underbrace{\theta_2,\ldots,\theta_2}_{d},\underbrace{\theta_1,\ldots,\theta_1}_d).\]

    The block factors are \[\textbf{A}:=\textbf{QT},\qquad B:= (\textbf{I}_d-\textbf{Q})\]

    Then \[\textbf{I}_{2d}-\textbf{B}_g\Delta = \begin{pmatrix}
        \textbf{I}_d-\textbf{A}\theta_2 & -\textbf{B}\theta_1\\[4pt]-\textbf{A}\theta_2&\textbf{I}_d-\textbf{B}\theta_1
    \end{pmatrix}\] Using $\textbf{u}:=(\textbf{I}_{2d}-\textbf{B}_g\Delta)^{-1}\textbf{b}=\begin{pmatrix}
        \textbf{u}_1\\\textbf{u}_2
    \end{pmatrix}$, the system again simplifies to $\textbf{u}_1=\textbf{u}_2=:\textbf{u}_d$. Adding (\eqref{eq:block1} and \eqref{eq:block2} from proof 1) gives 

    \begin{align*}&2\textbf{I}_d\textbf{u}_d-2\textbf{A}\theta_2\textbf{u}_d-2\textbf{B}\theta_1\textbf{u}_d=2\textbf{Qt}\\[4pt]&\Longleftrightarrow (\textbf{I}_d-\textbf{A}\theta_2-\textbf{B}\theta_1)\textbf{u}_d=\textbf{Qt}\end{align*}
    
    Using the definitions of \textbf{A} and \textbf{B},
    \begin{align*}&(\textbf{I}_d-\textbf{QT}\theta_2-(\textbf{I}-\textbf{Q})\theta_1)\textbf{u}_d=\textbf{Qt}\\[4pt]\Longleftrightarrow & \textbf{u}_d=(\textbf{I}_d-(\textbf{QT}\theta_2+(\textbf{I}-\textbf{Q})\theta_1))^{-1}\textbf{Qt}.\end{align*}

    Restating $G_Y$ gives \begin{align*}
        G_Y(\theta_1,\theta_2)&=\boldsymbol\beta^\top\Delta\textbf{u}=\boldsymbol\pi^\top\theta_2\textbf{u}_d=\boldsymbol\pi^\top\theta_2(\textbf{I}_d-(\textbf{QT}\theta_2+(\textbf{I}-\textbf{Q})\theta_1))^{-1}\textbf{Qt}.
    \end{align*}

    The PGF for \textbf{X} relies on the pathwise inclusion of a geometric kernel for the per-transition emission $k$, as well as a certainly contributing 1 to $\tau$ per step: \[G_{k,s}^i(\theta_1,\theta_2)=\sum\limits_{k\ge0,s=1}(1-q_i)^kq_i\theta_1^k\theta_2=\frac{q_i\theta_2}{1-(1-q_i)\theta_1}.\]
    Thus, the PGF components can be collected to form the matrix \textbf{D}:  \[\textbf{D}(\theta_1,\theta_2):=\text{diag}\left(\frac{q_i\theta_2}{1-(1-q_i)\theta_1}\right)_{i=1}^d=\textbf{Q}\theta_2(\textbf{I}-(\textbf{I}-\textbf{Q})\theta_1)^{-1}\] As the dependence on $\theta_1,\theta_2$ is clear, it will be omitted henceforth. The remaining calculation is as in the Bernoulli case: \[G_X(\theta_1,\theta_2)=\boldsymbol\pi^\top\textbf{D}\sum\limits_{n\ge0}(\textbf{TD})^n\textbf{t}=\boldsymbol\pi^\top\textbf{D}(\textbf{I}_d-\textbf{TD})^{-1}\textbf{t}.\]

    It remains to show that \[\boldsymbol\pi^\top\theta_2(\textbf{I}_d-(\textbf{QT}\theta_2+(\textbf{I}-\textbf{Q})\theta_1))^{-1}\textbf{Qt} = \boldsymbol\pi^\top\textbf{D}(\textbf{I}_d-\textbf{TD})^{-1}\textbf{t}.\]
    Expanding the expression for D again gives 

    \[
        \boldsymbol\pi^\top\theta_2(\textbf{I}-(\textbf{QT}\theta_2+(\textbf{I}-\textbf{Q})\theta_1))^{-1}\textbf{Qt} = \boldsymbol\pi^\top\textbf{Q}\theta_2(\textbf{I}-(\textbf{I}-\textbf{Q})\theta_1)^{-1}(\textbf{I}-\textbf{T}\textbf{Q}\theta_2(\textbf{I}-(\textbf{I}-\textbf{Q})\theta_1)^{-1})^{-1}\textbf{t}.\]

        One can immediately get rid of the leading $\theta_2\boldsymbol\pi^\top$ as well as the trailing \textbf{t}, leaving
        \[(\textbf{I}-(\textbf{QT}\theta_2+(\textbf{I}-\textbf{Q})\theta_1))^{-1}\textbf{Q} = \textbf{Q}(\textbf{I}-(\textbf{I}-\textbf{Q})\theta_1)^{-1}(\textbf{I}-\textbf{T}\textbf{Q}\theta_2(\textbf{I}-(\textbf{I}-\textbf{Q})\theta_1)^{-1})^{-1}.\]

    Define $\textbf{A}:=\textbf{I}-(\textbf{I}-\textbf{Q})\theta_1$ and $\textbf{B}:=-\textbf{T}\theta_2$, so that the equation becomes

    \[(\textbf{A}+\textbf{QB})^{-1}\textbf{Q}=\textbf{QA}^{-1}(\textbf{I}+\textbf{BQA}^{-1})^{-1}.\]

    We use a special case of the well-known identity by Woodbury \cite{woodburry} 
    
    \[(\textbf{A}+\textbf{UV})^{-1}=\textbf{A}^{-1}-\textbf{A}^{-1}\textbf{U}(\textbf{I}+\textbf{VA}^{-1}\textbf{U})^{-1}\textbf{VA}^{-1},\] transforming the equation to

    \[(\textbf{A}^{-1}-\textbf{A}^{-1}\textbf{Q}(\textbf{I}+\textbf{BA}^{-1}\textbf{Q})^{-1}\textbf{BA}^{-1})\textbf{Q} = \textbf{QA}^{-1}(\textbf{I}+\textbf{B}\textbf{Q}\textbf{A}^{-1})^{-1}.\]

    Notably, \textbf{A}, \textbf{Q}, and \textbf{I} are diagonal matrices and thus commute. For \textbf{A} this is true because it is the sum of diagonal matrices. Define $\textbf{C}:=(\textbf{I}+\textbf{BQA}^{-1})^{-1}$. The identity can then be shown using simple matrix algebra:
    \begin{align*}&(\textbf{A}^{-1}-\textbf{A}^{-1}\textbf{QCB}\textbf{A}^{-1})\textbf{Q} = \textbf{Q}\textbf{A}^{-1}\textbf{C}&&\text{(Distributive \& Commutativity)}\\[4pt]
       &(\textbf{QA}^{-1}-\textbf{A}^{-1}\textbf{QCBQA}^{-1}) = \textbf{Q}\textbf{A}^{-1}\textbf{C}&&|(\textbf{QA}^{-1})^{-1}\text{ from the left}\\[4pt]
       &\textbf{I}-\textbf{CBQA}^{-1} = \textbf{C}&&|+\textbf{CBQA}^{-1}\\[4pt]
&\textbf{I} = \textbf{C}+\textbf{CBQA}^{-1}&&|\textbf{C}^{-1}\text{ from the left}\\[4pt]
&\textbf{C}^{-1} = \textbf{I}+\textbf{BQA}^{-1}&&|\textbf{C}\text{ left and }(\textbf{I}+\textbf{BQA}^{-1})^{-1}\text{ right}\\[4pt]
    & (\textbf{I}+\textbf{BQA}^{-1})^{-1}=\textbf{C},\end{align*}
    which is true by definition.  
\end{proof}

\section{The EM algorithm for the IEM}\label{inferenceAppendix}

We derive the EM updates for the geometric-reward IEM by specializing the general MDPH procedure of He and Ren \cite{He2016Sep} to the special case of RRDPH, which has non-overlapping rewards in $\{0,1\}$. Non-IEM-based RRDPH applications can use an analogous EM algorithm. 

\paragraph{E-Step.}
Algorithm \ref{expAlg} gives step-by-step instructions to recover the expected state transition matrix conditioned on parameters $\nu^{cur}$, $\eta^{cur}$, and $\textbf{q}^{cur}$. The expected number of transitions to absorption is then added as column 2d+1 to the rest of the expected transition counts, making the E-step at iteration $i$ return a $2d\times2d+1$ matrix denoted $\widehat{\textbf{N}}^{(i)}$.

\begin{algorithm}\caption{Gathering the Expected Total Transition Counts}\label{expAlg}
    \begin{algorithmic}[1]\vspace{2mm}
        \STATE Input $\bar\nu,\;\bar\eta$, and $\bar{\textbf{q}}$
        \STATE Initiate $C^0=\textbf I_2$ \vspace{3mm}
        \STATE Calculate $\textbf{T}_{IE}$ from $\bar\nu,\;\bar\eta,$ and $\bar{\textbf{q}}$ using eq. \eqref{ieMod}, $\textbf{b}_0=(\textbf I_{2d} - \textbf{T}_{IE}) \textbf{e}$\vspace{3mm}
        \STATE For each reward type $1\le r\le 2$, define $\textbf{B}_r\in\mathbb R^{d\times d}$ with \[(\textbf{B}_r)_{i,j} = \begin{cases}
                (\textbf{T}_{IE})_{i,j},& \text{if  } R_{j,r}=1\\
    0,              & \text{otherwise}
        \end{cases}\]\vspace{3mm}
        \STATE Set  $\textbf{p}_\textbf{Y}((0,0))=\textbf{b}_0$, $\boldsymbol\alpha((0,0))=\boldsymbol\pi$\vspace{3mm}
        \STATE Update $\textbf{p}_\textbf{Y}(y)$ and $\boldsymbol\alpha(y)$:\[\textbf{p}_\textbf{Y}(y)=\sum\limits_{i=1}^2\textbf{B}_i\cdot\textbf{p}_\textbf{Y}(y-(\textbf{C}^0)_i)\cdot\textbf{1}(\forall_{j\in\{1,2\}}: y_j \ge(\textbf{C}^0)_{i,j}),\]
        
        and \[\boldsymbol\alpha(y)=\sum\limits_{i=1}^2\boldsymbol\alpha(y-(\textbf{C}^0)_i)\cdot\textbf{B}_i\cdot\textbf{1}(\forall_{j\in\{1,2\}}: y_j \ge(\textbf{C}^0)_{i,j}),\]
        where $(\textbf{C}^0)_i$ is the i-th row of $\textbf{C}^0$.
        \STATE For observations $1\le i\le n$ and observation type $1\le r\le 2$ we now know the expected transitions conditioned on observations and parameters: \[\mathbb E(N_{(i,j),r}\mid y_i,\boldsymbol\theta,(\textbf{P/Q})) = (\boldsymbol\pi\cdot\textbf{p}_\textbf{Y}(y_i))^{-1}\sum\limits_{u:u\le y_i-\textbf{C}^0_r}\boldsymbol\alpha_i(u)(\textbf{B}_{r})_{(ij)}\textbf{p}_{\textbf{Y},j}(y_i-\textbf{C}^0_r-u),\] as well as the conditional expected absorptions \[\mathbb E(N_{(i,0)}\mid y_i,\boldsymbol\theta,(\textbf{P/Q})) = (\boldsymbol\pi\cdot\textbf{p}_\textbf{Y}(y_i))^{-1}\boldsymbol\alpha_i(y_i)(\textbf{b}_0)_i\]
    \end{algorithmic}
\end{algorithm}

\paragraph{M-Step.}
Let \(\mathcal S^\star=\{1,\dots,2d\}\) denote the expanded transient state space, and for each original severity state \(j\in\{1,\dots,d\}\), define
\[
\mathcal I_j=\{j,j+d\}.
\]
Thus, \(\mathcal I_j\) contains the rewarded and unrewarded copy of the original state \(j\). For observation \(i\), let
\[
N^{(i)}_{ab}, \qquad a,b\in\{1,\dots,2d\},
\]
denote the latent number of transitions from expanded state \(a\) to expanded state \(b\), and let \(N^{(i)}_{a,2d+1}\) denote the number of transitions from expanded state \(a\) to absorption. The E-step computes the conditional expectations
\[
\widehat N^{(i)}_{ab}
=
\mathbb E\!\left[N^{(i)}_{ab}\mid Y_i,\theta^{(m)}\right],
\qquad
\widehat N^{(i)}_{a,2d+1}
=
\mathbb E\!\left[N^{(i)}_{a,2d+1}\mid Y_i,\theta^{(m)}\right].
\]
For compactness, define grouped counts by
\[
\widehat N^{(i)}_{\mathcal A,\mathcal B}
=
\sum_{a\in\mathcal A}\sum_{b\in\mathcal B}\widehat N^{(i)}_{ab},
\]
for subsets \(\mathcal A,\mathcal B\subseteq \mathcal S^\star\). In particular, \(\widehat N^{(i)}_{\mathcal I_j,\mathcal I_k}\) is the expected number of transitions from original state \(j\) to original state \(k\), regardless of whether the transition occurs in the rewarded or unrewarded copy.

Let
\[
\widehat N_{jk}=\sum_{i=1}^n \widehat N^{(i)}_{\mathcal I_j,\mathcal I_k},
\qquad
\widehat N_{j0}=\sum_{i=1}^n \widehat N^{(i)}_{\mathcal I_j,2d+1}.
\]
These are the aggregated expected transition counts used in the M-step.

The expected complete-data log-likelihood separates into a transition part and a reward part:
\[
\ell_c(\nu,\eta,\mathbf Q)=\ell_T(\nu,\eta)+\ell_Q(\mathbf Q),
\]
 Writing
\[
S_{\mathrm{stay}}=\sum_{j=1}^d \widehat N_{jj},\qquad
S_{\mathrm{up}}=\sum_{j=1}^{d-1}\widehat N_{j,j+1}+\widehat N_{d0},\qquad
S_{\mathrm{down}}=\widehat N_{10}+\sum_{j=2}^d \widehat N_{j,j-1},
\]
the transition part becomes
\[
\ell_T(\nu,\eta)
=
S_{\mathrm{stay}}\log \nu
+
S_{\mathrm{up}}\log\!\bigl((1-\nu)\eta\bigr)
+
S_{\mathrm{down}}\log\!\bigl((1-\nu)(1-\eta)\bigr).
\]
Maximizing this expression gives the closed-form updates
\begin{equation}\label{thetaMax}
\widehat\nu=
\frac{S_{\mathrm{stay}}}{S_{\mathrm{stay}}+S_{\mathrm{up}}+S_{\mathrm{down}}},
\qquad
\widehat\eta=
\frac{S_{\mathrm{up}}}{S_{\mathrm{up}}+S_{\mathrm{down}}}.
\end{equation}

For the geometric rewards, define
\[
F_j=\sum_{i=1}^n \widehat N^{(i)}_{\mathcal I_j,\,j+d},
\qquad
U_j=\sum_{i=1}^n \sum_{b=1}^{2d+1}\widehat N^{(i)}_{\mathcal I_j,b}.
\]
Here \(F_j\) is the expected number of reward failures in original state \(j\), and \(U_j\) is the total expected number of reward opportunities in that state. The reward part of the complete-data log-likelihood is
\[
\ell_Q(\mathbf Q)
=
\sum_{j=1}^d \Bigl[(U_j-F_j)\log q_j + F_j\log(1-q_j)\Bigr],
\]
which yields the update
\begin{equation}\label{rewardMaxxing}
\widehat q_j
=
1-\frac{F_j}{U_j}
=
\frac{U_j-F_j}{U_j},
\qquad j=1,\dots,d.
\end{equation}

These formulas are the specialized MDPH EM updates of \cite{He2016Sep} for the IEM representation used in this paper.

\paragraph{Regression version.}
If the reward probabilities are parameterized by the linear model
\[
\logit(q_j)=\beta_{q,0}+\beta_{q,1}j,
\qquad j=1,\dots,d,
\]
then the update for \(\mathbf Q\) is replaced by a weighted quasibinomial GLM with response counts \(U_j-F_j\), totals \(U_j\), and covariate \(j\). The regression updates for \(\nu\) and \(\eta\) are obtained analogously from weighted quasibinomial GLMs using the corresponding pseudo-observations and weights (denominators of the maximization function) produced by the E-step.

\paragraph{Complete Algorithm.} Algorithm \ref{maxStepAlg} describes the entire EM method, where the expected steps $\textbf{N}^{(t)}$ are found using algorithm \ref{expAlg}, and the parameter maximizations use equations \eqref{thetaMax} and \eqref{rewardMaxxing}.

\begin{algorithm}\caption{Expectation-Maximization Iteration for the IEM}
    \begin{algorithmic}[1]\label{maxStepAlg}\vspace{2mm}
        \STATE Initialize $\boldsymbol\theta_0,\;\textbf{Q}_0;\; t=0$, maximal iterations $\max_{iter}$ and minimal variation $\min_{var}$\vspace{2mm}
        \STATE Find the conditional expected steps $\textbf{N}^{(t)}$ using algorithm \ref{expAlg}.
        \STATE Maximize the transition parameters: \[\boldsymbol\theta^{(t+1)}=(\widehat{\nu(\textbf{N}^{(t)})},\widehat{\eta(\textbf{N}^{(t)}))}\]\vspace{-4mm}
        \STATE Maximize the reward probabilities: \[diag(\textbf{Q}^{(t+1)}) = (\widehat{q_1}(\textbf{N}^{(t)}),\cdots,\widehat{q_d}(\textbf{N}^{(t)}))\]\vspace{-4mm}
        \STATE Set $t:=t+1$\vspace{2mm}
        \STATE If $t<\max_{iter}$ and $\min_{var}<\mid\mid\boldsymbol\theta_{t}-\boldsymbol\theta_{t-1}\mid\mid$, go to step 2\vspace{2mm}
        \STATE Return ($\theta_{t},\textbf{P}_{t}/\textbf{Q}_{t}$)\vspace{1mm}
    \end{algorithmic}
\end{algorithm}

\end{appendices}


\bibliography{sn-bibliography}

\end{document}